\documentclass{llncs}
\usepackage{natbib} 
\usepackage{pstricks}
\usepackage{pstricks-add}
\usepackage{graphicx}
\usepackage{amsmath}
\usepackage{hyperref}
\usepackage{amssymb}
\usepackage{multirow}
\usepackage{booktabs}
\usepackage{tabularx}
\usepackage{bm}
\begin{document}
\renewcommand{\thepage}{\arabic{page}}
\pagestyle{plain} 
\setcounter{page}{1}
\title{A Least Squares Estimation of a Hybrid log-Poisson Regression and its Goodness of Fit for Optimal Loss Reserves in Insurance}
\titlerunning{Hamiltonian Mechanics}  
%
\author{WOUNDJIAGU\'{E} Apollinaire \inst{1,4} \and Mbele Bidima Martin Le Doux\inst{2}
\and Waweru Mwangi Ronald \inst{3}}
\authorrunning{Apollinaire WOUNDJIAGU\'{E} et al.} 
%
\tocauthor{Apollinaire WOUNDJIAGU\'{E}, Martin Le Doux Mbele Bidima and Ronald Waweru Mwangi}
\institute{Institute of Basic Sciences Technology and Innovation, Pan African University-Jomo Kenyatta University of Agriculture and Technology, P.O. Box 62000-00200, Nairobi-Kenya,\\
\email{appoappolinaire@yahoo.fr},\\
\texttt{woundjiague.apollinaire@students.jkuat.ac.ke}
\and
Faculty of Science, University of Yaounde I, P.O.Box 812 Yaoundé, Cameroon
\and 
School of Computing and Information Technology, Jomo Kenyatta University of Agriculture and Technology, P.O. Box 62000-00200, Nairobi-Kenya
\and National Advanced School of Engineering, University of Maroua, PO.Box 46 Maroua, Cameroon}
\maketitle   
\begin{abstract}
 In this article, the parameters of a hybrid log-linear model (log-Poisson) are estimated using the fuzzy least-squares (FLS) procedures (Celmiņ\v{s}, 987a,b; D'Urso and Gastaldi, 2000; DUrso and Gastaldi, 2001). A goodness of fit have been derived in order to assess and compare this new model and the classical log-Poisson regression in loss reserving framework (Mack, 1991). Both the hybrid model and its goodness of fit are performed on a loss reserving data.

\end{abstract}
\section{Introduction}
An important role of a non-life actuary is the calculation of provisions, mainly Incurred But Not Reported reserve (IBNR). Then, finding the fair value of loss reserve is a relevant topic for non-life actuaries. Indeed, insurance companies must simultaneously have enough reserves to meet their commitment to policyholders and have enough funds for their investments. Therefore several methods have been proposed in actuarial science literature to capture this fair value.\par
In one hand, we distinguish deterministic methods (Bornhuetter and Ferguson, 1972; Taylor,
1986; Linnemann, 1984). They provide crisp predictions for reserves. In the other hand, Taylor et al.
(2003); W\"{u}thrich and Merz (2008); Mack (1991); England and Verrall (2002) present stochastic methods. Those methods don't give only a crisp value of the reserves but provide also their variability. But even stochastic methods have weakness.\par
In Straub and Swiss (1988), there are some experiences where stochastic methods can give unrealistic estimates. For example, when the claims are related to body injures, the future losses for the company will depend on the growth of the wage index that help to determine the amount of indemnity, and depends also on changes in court practices and public awareness of liability matters. Then the information is vague. Therefore the use of Fuzzy Set Theory becomes very attractive when the information is vague as in this case. \par 
de Andr\'{e}s S\'{a}nchez (2006); de Andr\'{e}s-S\'{a}nchez (2007, 2012); de Andr\'{e}s S\'{a}nchez (2014) present the interest of fuzzy regression models (FRM) in the calculus of loss reserves in insurance using the concept of expected value of a Fuzzy Number (FN) (de Campos Ib\'{a}\~{n}ez and Mu\~{n}oz,
1989). Asai (1982) is the first to develop a FRM where the coefficient are fuzzy numbers (Dubois
and Prade, 1988). In the case of loss reserving, FN are easy to handle arithmetically unlike in the case of classical regression where the coefficients are random variables and are not easy to handle arithmetically. Another difference between fuzzy regression and classical regression is in dealing with errors as fuzzy variables in fuzzy regression modelling while errors are considered as random residuals in classical regression. But to integrate both fuzziness and randomness into a regression model, one should think about hybrid regression models.\par
Then we have developed in our previous article a hybrid log-Poisson regression inspired from the FRM (Asai, 1982) and taking into account an optimize $h-$value in the linear program. However, the fuzzy parameters in this model are calculated through an optimization program and does not provide an explicit form of the estimated parameters.\par
In this article, we derive the exact form of the estimated fuzzy parameters of the hybrid log-Poisson regression by using the concept of fuzzy least-squares (Celmiņ\v{s}, 987a,b; D'Urso and
Gastaldi, 2000; DUrso and Gastaldi, 2001). We develop a goodness of fit index to assess this new model and to compare it with the classical one. A numerical application on a loss reserving data will be provided.\par
The article is organized as follows: We first present some definitions and preliminaries concepts in fuzzy logic as the first section; In the second section, we derive a least squares estimation of the log-Poisson Model; The new estimation procedure for the hybrid log-Poisson regression is developed in section three; In the fourth and fifth section, a goodness of fit and the implementation algorithm of the hybrid model are respectively developed. Finally, a numerical example and a conclusion are suggested at the end of this paper.
\section{Preliminaries on Fuzzy Sets}
In this section, we review some concepts related to our research. That is the concept of fuzzy set, membership function, fuzzy number, weighted function of FN, fuzzy linear regression estimation according to least square approach.
\subsection{Review on some Definitions and Basic Properties of Fuzzy Sets}
\begin{definition}(Zadeh, 1965)\\
	Let $\Omega$ be a non empty set and $x \in \Omega$. In classical set theory, a subset $A$ of $\Omega$ can be defined by its characteristic function $\chi_{A}$ as a mapping from the elements of $\Omega$ to the elements of the set $\{0, 1\}$ ,
	\begin{equation}
	\chi_{A} : \Omega \longrightarrow \{ 0, 1\} 
	\end{equation}
	This mapping may be represented as a set of ordered pairs, with exactly one ordered pair present for each element of $\Omega$. The first element of the ordered pair is an element of the set $\Omega$, and the second element is an element of the set$\{0, 1\}$ . The value zero is used to represent non-membership, and the value one is used to represent membership. The truth or falsity of the statement "$x$ is in $A$" is determined by the ordered pair $(x, \chi_{A}(x) )$. The statement is true if the second element of the ordered pair is 1, and the statement is false if it is 0.\par
	Similarly, a fuzzy subset (also called fuzzy set) $\tilde{A}$ of a set $\Omega$ can be defined as a set of ordered pairs, each with the first element from $\Omega$, and the second element from the interval $[0, 1]$, with exactly one ordered pair present for each element of $\Omega$. This defines a mapping called membership function. 
\end{definition}

\begin{definition} (Zadeh, 1965)\\
	The membership function of a fuzzy set $\tilde{A}$, denoted by $\mu_{\tilde{A}}$ is defined by
	\begin{equation}
	\mu_{\tilde{A}} : \Omega \longrightarrow [0,1]
	\end{equation}
	where $\mu_{\tilde{A}}$ is typically interpreted as the membership degree of element $x$ in the fuzzy set $\tilde{A}$.\par
	The degree to which the statement " $x$ is in $\tilde{A}$" is true is determined by finding the ordered pair $(x,\mu_{\tilde{A}}(x))$. The degree of truth of the statement is the second element of the ordered pair. A fuzzy set (Zadeh, 1965) $\tilde{A}$ on $\Omega$ can also be defined as a set of tuples:
	\begin{equation}
	\tilde{A} = \{(x,\mu_{\tilde{A}}(x)) \mid x \in \Omega \}.
	\end{equation}
	and could be represented by a graphic.
\end{definition}

\begin{definition}(Dubois and Prade, 1978)\\
	Let $\Omega$ be the set of objects and $\tilde{A} \subset\Omega.$ The $\alpha-$cut $\tilde{A}_{\alpha}$ of $\tilde{A}$ is the set defined by
	\begin{equation*}
	\tilde{A}_{\alpha} = \{x \in \Omega, \mu_{\tilde{A}}(x) \geqslant \alpha \}.  
	\end{equation*}	
\end{definition}

\begin{definition} (Dubois and Prade, 1988)
	\begin{enumerate}
		\item 
		A fuzzy number $\tilde{A}$ is a fuzzy set of a universe $\Omega$ (the real line $\mathbb{R}$) such that :
		\begin{enumerate}
			\item[a.] 
			all its $\alpha-$cut are convex which is equivalent to $\tilde{A}$ is	convex, that is $\forall x_{1}, x_{2} \in \mathbb{R}$ and $\lambda \in [0,1], \quad \mu_{\tilde{A}}(\lambda x_{1} + (1-\lambda) x_{2}) \geqslant \min ( \mu_{\tilde{A}}(x_{1}),\mu_{\tilde{A}}(x_{2}) )$;
			\item[b.]
			$\tilde{A}$ is	normalized, that is $\exists x_{0} \in \Omega$ such that $\mu_{\tilde{A}}(x_{0})=1.$
			\item[c.]
			$\mu_{\tilde{A}}$ is continued membership function of bounded support, where $\Omega=\mathbb{R}$ and $[0,1]$ are equipped with the natural topology.	
		\end{enumerate}
		\item
		A triangular fuzzy number (TFN) $\tilde{\gamma}$ is a fuzzy number denoted by $\tilde{\gamma}=(\beta^{L}, \alpha^{c},\beta^{R}); \quad \beta^{L}, \alpha^{c},\beta^{R}\in \mathbb{R}$, such that $\mu_{\tilde{A}}(\beta^{L})=\mu_{\tilde{A}}(\beta^{R})=0$ and $\mu_{\tilde{A}}(\alpha^{c})=1$ with $\alpha^{c}$ the centre of $\tilde{\gamma}$, $\beta^{L}$ its left spread and $\beta^{R}$ its right spread (Lai and Hwang, 1992).\par
		A TFN $\tilde{\gamma}$ could be defined with its membership degree function $\mu_{\tilde{\gamma}}$ or, with its $h-$level ($\alpha-$ cut  ($h\in [0,1]$) $\gamma_{h}$ (see Dubois and Prade (1988)), i.e
		\begin{equation}\label{57}
		\mu_{\tilde{\gamma}}(x)= \left\{ \begin{array}{cl}
		1 - \dfrac{\alpha^{c} - x}{\beta^{L}} & \text{if}\quad \alpha - \beta^{L} < x \leqslant \alpha\\
		1 - \dfrac{x - \alpha^{c}}{\beta^{R}} & \text{if}\quad \alpha < x \leqslant \alpha +  \beta^{R}\\
		0 & \text{if} \quad \text{otherwise}
		\end{array}\right. 
		\end{equation}
		or
		\begin{equation}\label{58}
		\tilde{\gamma}_{h}=[\gamma_{L_{h}} , \gamma_{R_{h}}]=\big [\alpha^{c} - \beta^{L}(1-h), \alpha^{c} + \beta^{R} (1-h)  \big ]
		\end{equation}
		\begin{itemize}
			\item 
			If $\alpha^{c} - \beta^{L} = \beta^{R} - \alpha^{c}$, then $\tilde{\gamma}$ define a STFN
			\item
			Otherwise $\beta^{L} \neq \beta^{R}$, then $\tilde{\gamma}$ define an ATFN (see Figure \ref{fig2}).\\ 
			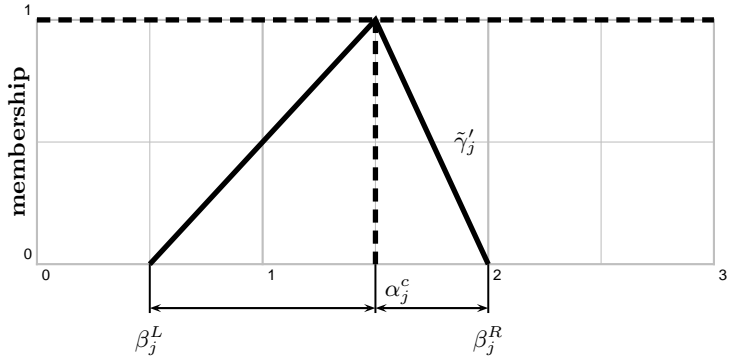
\begin{figure}[h!]
				\begin{center}
					\caption{Asymmetric triangular fuzzy coefficients $\tilde{\gamma}_{j}=(\beta_{j}^{L},\alpha^{c}_{j},\beta_{j}^{R})$}\label{fig2}
					\bigskip
					\psset{xunit=3cm,yunit=3.25cm}
					\begin{pspicture}(0,0)(3,1)
					\psgrid[gridcolor=lightgray,subgriddiv=2,%
					subgridcolor=lightgray,gridlabels=2mm](0,0)(3,1)
					\psline[linewidth=2pt](0.5,0)(1.5,1)(2,0)
					\psline[linewidth=2pt,linestyle=dashed](0,1)(3,1)
					\psline[linewidth=2pt,linestyle=dashed](1.5,1)(1.5,0)
					\rput{90}(-0.08,0.5){\textbf{membership}}
					\rput{0}(1.9,0.5){$\tilde{\gamma}^{\prime}_{j}$}
					\psline[linewidth=0.75pt](0.5,0)(0.5,-0.2)
					\psline[linewidth=0.75pt](1.5,0)(1.5,-0.2)
					\psline[linewidth=0.75pt](2,0)(2,-0.2)
					\pcline[offset=-12pt]{<->}(0.5,-0.05)(1.5,-0.05)
					\pcline[offset=-12pt]{<->}(1.5,-0.05)(2,-0.05)
					\uput[-90](0.5,-0.2){$\beta_{j}^{L}$}
					\uput[45](1.5,-0.2){$\alpha^{c}_{j}$}
					\uput[-90](2,-0.2){$\beta_{j}^{R}$}
					\end{pspicture}
				\end{center}\end{figure}
				\bigskip
			\end{itemize}
		\end{enumerate}
	\end{definition}
\paragraph*{Notes and Comments.}
It is well know that if $\tilde{A}$ is a fuzzy number, then $\tilde{A}_{h}$, the $h$ level ($\alpha-$cut) of $\tilde{A}$ is a compact set of $\mathbb{R}$, for all $h \in [0,1].$
\subsection{Review on Celmi{\c{n}}{\v{s}} Least Squares Model for Fuzzy Linear Regression}
In this subsection, we present the least-squares approach to estimated the fuzzy linear regression (Celmi\c{n}\v{s}, 987a,b; D'Urso and Gastaldi, 2000; DUrso and Gastaldi, 2001) rather than the possibilistic approach (Asai, 1982; Ishibuchi and Nii, 2001).\par
Let us consider crisp explanatory variables $x_{ij} (i=1,\ldots,n; j=1,\ldots,p)$ describing triangular fuzzy variables $\tilde{Y}_{i}=(Y^{L}_{i},Y^{c}_{i},Y^{R}_{i})$, where $Y^{c}_{i}$ is the centre of $\tilde{Y}_{i}$, $Y^{L}_{i}$ and $Y^{R}_{i}$ its left and right spreads respectively.\par
In matrix form, the fuzzy linear regression model between $\mathbf{X}$ (matrix of explanatory variables $x_{ij}$) and $\tilde{\mathbf{Y}}$ ( vector of dependent variables $\tilde{Y}_{i}$) can be written as following :
\begin{equation}\label{4}
\left\{ \begin{array}{cl}
Y^{c} & = Y^{c^{\ast}} + \varepsilon\\
Y^{L} & = Y^{L^{\ast}} + \xi \\
Y^{R} & = Y^{R^{\ast}} + \eta
\end{array}\right.
\end{equation}
where
\begin{equation}\label{5}
\left\{ \begin{array}{cl}
Y^{c^{\ast}} & = \mathbf{X}\bm{\beta}\\
Y^{L^{\ast}} & = Y^{c^{\ast}}\theta + \mathbf{1}\lambda \\
Y^{R^{\ast}} & = Y^{c^{\ast}}\delta + \mathbf{1}\mu
\end{array}\right.
\end{equation}
and 
\begin{description}
	\item[$\mathbf{X} :$] is a $n\times (p+1)$ matrix containing the vector $\mathbf{1}$ concatenated to $p$ crisp input variables.
	\item[$\bm{\beta} :$] is a $(p+1)\times 1$ vector of regression parameters for the regression model for $Y^c$.
	\item[$Y^{c} :$] is a $n\times 1$ vector of the observed centres.
	\item[$Y^{c^{\ast}} :$] is a $n\times 1$ vector of interpolated centres.
	\item[$Y^{L} :$] is a $n\times 1$ vector of observed left spreads.
	\item[$Y^{L^{\ast}} :$] is a $n\times 1$ vector of observed interpolated left spreads.
	\item[$Y^{R} :$] is a $n\times 1$ vector of observed right spreads.
	\item[$Y^{R^{\ast}} :$] is a $n\times 1$ vector of observed interpolated right spreads.
	\item[$\theta , \lambda, \delta , \mu :$] are the regression parameters for the regression model for $Y^{L}$ and $Y^{R}$.
	\item[$\mathbf{1} :$] $n\times 1$ vector of ones.
	\item[$\varepsilon, \xi, \eta :$] are $n\times 1$ vector of error terms.
\end{description}
\begin{remark}
Model \eqref{4} is based on 3 sub-models. The first interpolate the centre of $\tilde{Y}_{i}$. The two others are built over the first and yield the spreads.
\end{remark}
\begin{theorem}
The iterative least squares estimates $\hat{\bm{\beta}}, \hat{\theta}, \hat{\lambda}, \hat{\delta}$ and $\hat{\mu}$ of the parameters $\bm{\beta}, \theta, \lambda, \delta$ and $\mu$ in the system \eqref{4} are given by:
\begin{align}
\hat{\bm{\beta}} & = \dfrac{1}{\big( 1 + \theta^{2} + \delta^{2}\big)}\bigg[ (\mathbf{X}^{T}\mathbf{X})^{-1} \mathbf{X}^{T}\big( Y^{c} + (Y^{L} - \displaystyle\mathbf{1}\lambda)\theta + (Y^{R} - \displaystyle\mathbf{1}\mu)\delta \big)\bigg]\\
\hat{\theta} & = \big( \bm{\beta}^{T}\mathbf{X}^{T}\mathbf{X} \bm{\beta}\big)^{-1}\big( \bm{\beta}^{T}\mathbf{X}^{T}Y^{L} - \bm{\beta}^{T}\mathbf{X}^{T}\mathbf{1}\lambda\big)\\
\hat{\lambda} & = \dfrac{1}{n}\bigg( (Y^L)^{T} \mathbf{1} - \bm{\beta}^{T}\mathbf{X}^{T}\mathbf{1}\theta\bigg)\\
\hat{\delta} & = \big( \bm{\beta}^{T}\mathbf{X}^{T}\mathbf{X} \bm{\beta}\big)^{-1}\big( \bm{\beta}^{T}\mathbf{X}^{T}Y^{R} - \bm{\beta}^{T}\mathbf{X}^{T}\mathbf{1}\mu\big)\\
\hat{\mu} & = \dfrac{1}{n}\bigg( (Y^R)^{T} \mathbf{1} - \bm{\beta}^{T}\mathbf{X}^{T}\mathbf{1}\delta\bigg)
\end{align}
\end{theorem}
\begin{proof}
(see D'Urso (2003)).
\end{proof}
\section{A Least Squares Estimation of the log-Poisson Model}\label{s1}	
In this section, we provide a least squares estimation for the classical log-Poisson regression (Mack,
1991) in loss reserving framework.\par
Let Table \ref{tab1} be a \textit{run-off triangle}, where $Y_{ij}$ is the total loss regarding the underwriting period $i$ which have been paid with $j$ periods delay and the loss amounts $Y_{ij}$ with $i+j=k$ have been paid in calendar year $k\in \mathbb{N}$.\par
\begin{table}[!h]
	\centering
	\begin{tabular}{lclclclclcl}
		\toprule
		\multicolumn{8}{c}{\bf Development Year}\\ 
		& & \bf 0 & \bf 1 & \bf $\ldots$ & $l$ & $\ldots$ & $k-i$ & $\ldots$ & $k-1$ & $k$ \\ \hline
		\multirow{9}{6mm}{\rotatebox{90}{\bf Accident Year}} & \bf 0  & $Y_{0,0}$ & $Y_{0,1}$ & $\ldots$ & $Y_{0,l}$ & $\ldots$ & $Y_{0,k-i}$ & $\ldots$ & $Y_{0,k-1}$ & $Y_{0,k}$\\
		& \bf 1   & $Y_{1,0}$ & $Y_{1,1}$ & $\ldots$ & $Y_{1,l}$ & $\ldots$ & $Y_{1,k-i}$ & $\ldots$ & $Y_{1,k-1}$ \\
		& $\vdots$  & $\vdots$ & $\vdots$ & $\ldots$ & $\vdots$ & $\ldots$ & $\vdots$ \\
		& $i$   & $Y_{i,0}$ & $Y_{i,1}$ & $\ldots$ & $Y_{i,l}$ & $\ldots$ & $Y_{i,k-i}$\\
		& $\vdots$  & $\vdots$ & $\vdots$ & $\ldots$ & $\vdots$ \\
		& $k-l$  & $Y_{k-l,0}$ & $Y_{k-l,1}$ & $\ldots$ & $Y_{i,k-i}$\\
		& $\vdots$  & $\vdots$ & $\vdots$\\
		& $k-1$  & $Y_{k-1,0}$ & $Y_{k-1,1}$\\ 
		& $k$  & $Y_{k,0}$\\ \hline
	\end{tabular}
	\caption{Triangle of observed incremental payments}\label{tab1}
\end{table}	
$\{Y_{ij} : i=1,\ldots , k ; j=1,\ldots, k\}$ are assumed to be independent and $\log-$Poisson distributed (Mack, 1991), i.e 
\begin{equation}
Y_{ij} \sim \mathcal{P}(e^{\nu_{ij}})=\mathcal{P}(e^{\tau + \alpha_{i} + \gamma_{j}})
\end{equation} 
We suppose that $\{ Y_{ij} : i=1,\ldots, k; j=1,\ldots, k-i+1\}$ are modelled by
\begin{align}\label{1}
\nu_{ij} = \ln  \mathbb{E}(Y_{ij}) \Leftrightarrow & \mathbb{E}(Y_{ij}) = e^{X_{i}^{T} \bm{\beta}}\\ \nonumber
\Leftrightarrow & \ln  \mathbb{E}(Y_{ij})= X_{i}^{T} \bm{\beta} \\ \nonumber
\Leftrightarrow & \ln \mathbb{E}(\mathbf{Y})=\mathbf{X}\bm{\beta}\nonumber
\end{align}
where
\begin{description}
	\item[$\mathbf{X}_{(n\times p)}$] is the matrix of explanatory variables $x_{ij}$
	\item[$\mathbf{Y}_{(n\times 1)}$] is the vector of observations $Y_{ij}$
	\item[$\bm{\beta}_{(p\times 1)}$] is the vector of parameters $n=\frac{1}{2}k(k+1)$ and $p=2k-1.$
\end{description}
\textbf{\underline{Example}:} Let us consider a $3\times 3$ \textit{run-off triangle}.
Hence
\begin{equation*}
\ln \mathbb{E} \begin{pmatrix}
Y_{11}\\
Y_{21}\\
Y_{31}\\
Y_{12}\\
Y_{22}\\
Y_{13}
\end{pmatrix} = \begin{pmatrix}
1 & 0 & 0 & 0 & 0 \\
1 & 1 & 0 & 0 & 0 \\
1 & 0 & 0 & 1 & 0 \\
1 & 0 & 1 & 0 & 0 \\
1 & 1 & 1 & 0 & 0 \\
1 & 0 & 0 & 0 & 1 
\end{pmatrix}
\begin{pmatrix}
\tau \\
\alpha_{2}\\
\gamma_{2}\\
\alpha_{3}\\
\gamma_{3}
\end{pmatrix}
\end{equation*}
\eqref{1} can be written as 
\begin{equation}\label{2}
\mathbf{Y}=\bm{\varepsilon}e^{\mathbf{X}\bm{\beta}}
\end{equation}
where $\bm{\varepsilon}_{n\times 1}$ is the vector of errors terms $\varepsilon_{ij}$ such that $\varepsilon_{ij} \sim \mathcal{P}(1)$ and $\mathbb{E}(\varepsilon_{ij})=\mathbb{V}(\varepsilon_{ij})=1.$\par
\begin{equation}\label{3}
\eqref{2}\Rightarrow \ln (\mathbf{Y})=\ln (\bm{\varepsilon}) + \mathbf{X}\bm{\beta}	
\end{equation}
According to least squares method, we can estimate the vector of parameters $\beta$ by minimizing
\begin{align*}
S(\bm{\beta}) & =\big(\ln (\bm{\varepsilon})\big)^{T}\big(\ln (\bm{\varepsilon})\big)\\
                  & =\bigg( \ln (\mathbf{Y}) - \mathbf{X}\bm{\beta}\bigg)^{T}\bigg( \ln (\mathbf{Y}) - \mathbf{X}\bm{\beta}\bigg)\\
                  & =\bigg( \big[\ln (\mathbf{Y})\big]^{T} - \bm{\beta}^{T}\mathbf{X}^{T}\bigg)\bigg( \ln (\mathbf{Y}) - \mathbf{X}\bm{\beta}\bigg)\\
                  & = \big[\ln (\mathbf{Y})\big]^{T} \big[\ln (\mathbf{Y})\big] - \big[\ln (\mathbf{Y})\big]^{T}\mathbf{X}\bm{\beta} - \bm{\beta}^{T}\mathbf{X}^{T}\big[\ln (\mathbf{Y})\big] + \bm{\beta}^{T}(\mathbf{X}^{T} \mathbf{X})\bm{\beta} \\
                  & = \big[\ln (\mathbf{Y})\big]^{T} \big[\ln (\mathbf{Y})\big] - 2\big[\ln (\mathbf{Y})\big]^{T}\mathbf{X}\bm{\beta} + \bm{\beta}^{T}(\mathbf{X}^{T} \mathbf{X})\bm{\beta}
\end{align*}
\begin{proposition}
Let $Y_{ij}$ be the loss amounts underwriting period $i$ which have been paid with $j$ periods delay in a certain \textit{run-off triangle}.\par
We assume that
$\{ Y_{ij} : i=1,\ldots, k; j=1,\ldots, k-i+1\}$ are modelled by
\begin{align}
Y_{ij} & \sim \mathcal{P}(e^{\nu_{ij}})=\mathcal{P}(e^{\tau + \alpha_{i} + \gamma_{j}})\\\nonumber
\Leftrightarrow  & \nu_{ij} = \ln  \mathbb{E}(Y_{ij}) \Leftrightarrow \mathbb{E}(Y_{ij}) = e^{X_{i}^{T} \bm{\beta}}\\ \nonumber
\Leftrightarrow & \ln  \mathbb{E}(Y_{ij})= X_{i}^{T} \bm{\beta} \\ \nonumber
\Leftrightarrow & \ln \mathbb{E}(\mathbf{Y})=\mathbf{X}\bm{\beta}\\\nonumber
\Leftrightarrow  & \mathbf{Y}  =\bm{\varepsilon}e^{\mathbf{X}\bm{\beta}}\nonumber
\end{align}
where
\begin{description}
	\item[$\mathbf{X}_{(n\times p)}$] is the matrix of explanatory variables $x_{ij}$
	\item[$\mathbf{Y}_{(n\times 1)}$] is the vector of observations $Y_{ij}$
	\item[$\bm{\beta}_{(p\times 1)}$] is the vector of parameters $n=\frac{1}{2}k(k+1)$ and $p=2k-1.$
	\item[$ \bm{\varepsilon}_{n\times 1} $] is the vector of errors terms $\varepsilon_{ij}$ such that $\varepsilon_{ij} \sim \mathcal{P}(1)$ and $\mathbb{E}(\varepsilon_{ij})=\mathbb{V}(\varepsilon_{ij})=1.$
\end{description}
Then the least squares estimator of $\bm{\beta}$ is given by
\begin{equation}
\hat{\bm{\beta}}^{LS}=(\mathbf{X}^{T}\mathbf{X})^{-1}\big[\ln (\mathbf{Y})\big]^{T}\mathbf{X}.
\end{equation}
\end{proposition}
\begin{proof}
\begin{equation}
\min_{\bm{\beta}} S(\bm{\beta}) \Leftrightarrow \left\{ \begin{array}{cl}
\dfrac{\partial S(\bm{\beta})}{\partial \bm{\beta}} = 0 \\
\dfrac{\partial ^{2} S(\bm{\beta})}{\partial \bm{\beta}^2} > 0
\end{array}\right.
\end{equation}
We have 
\begin{equation}
S(\bm{\beta})=\big[\ln (\mathbf{Y})\big]^{T} \big[\ln (\mathbf{Y})\big] - 2\big[\ln (\mathbf{Y})\big]^{T}\mathbf{X}\bm{\beta} + \bm{\beta}^{T}(\mathbf{X}^{T} \mathbf{X})\bm{\beta}
\end{equation}
\begin{equation}
\Rightarrow \dfrac{\partial S(\bm{\beta})}{\partial \bm{\beta}} = -2 \big[\ln (\mathbf{Y})\big]^{T}\mathbf{X} + 2(\mathbf{X}^{T}\mathbf{X})\bm{\beta}.
\end{equation}
\begin{align}
\dfrac{\partial S(\bm{\beta})}{\partial \bm{\beta}} = 0 & \Leftrightarrow (\mathbf{X}^{T}\mathbf{X})\bm{\beta} = \big[\ln (\mathbf{Y})\big]^{T}\mathbf{X}\\
& \bm{\beta} = (\mathbf{X}^{T}\mathbf{X})^{-1}\big[\ln (\mathbf{Y})\big]^{T}\mathbf{X}.
\end{align}
But 
\begin{equation}
\dfrac{\partial ^{2} S(\bm{\beta})}{\partial \bm{\beta}^2} = 2(\mathbf{X}^{T}\mathbf{X}),
\end{equation}
which is a semi definite positive matrix. Hence 
\begin{equation}
\hat{\bm{\beta}}^{LS}=(\mathbf{X}^{T}\mathbf{X})^{-1}\big[\ln (\mathbf{Y})\big]^{T}\mathbf{X}.
\end{equation}
\end{proof} 
\section{A Fuzzy Least Squares Estimation of a Hybrid Log-Poisson Regression for Loss Reserving}
In this section, we present another way to estimate the parameters of the hybrid log-Poisson regression, which is the extension of the classical log-Poisson regression (Mack, 1991) in loss reserving framework.\par
(Mack, 1991) assumes that the incremental payments $Y_{ij}$ are log-Poisson distributed, i.e,
\begin{equation}
Y_{ij} \sim \mathcal{P}(e^{\nu_{ij}}) \Rightarrow \mathbb{E}(Y_{ij})=e^{\nu_{ij}}=\varphi_{ij}\; \forall (i,j)\in \{ i=1,\ldots, k \} \times \{j=1,\ldots, k-i+1\}.
\end{equation} 
We assume that uncertainty about in the run-off triangle is due both to fuzziness and randomness. We suppose then that $\tilde{Y}_{ij}=(Y_{ij}^{L}, Y_{ij}^{c},Y_{ij}^{R})$ is a fuzzy Poisson random variable (Buckley, 2006), i.e,
\begin{align*}
\big[\mathbb{E}_{F}(\tilde{Y}_{ij})\big]_{h} & = \{\sum\limits_{x=0}^{+\infty}xe^{-\varphi_{ij}}\dfrac{(\varphi_{ij})^{x}}{x!}\mid \varphi_{ij} \in [\tilde{Y}_{ij}]_{h} \}\\
& =\{  \varphi_{ij}\mid \varphi_{ij} \in [\tilde{Y}_{ij}]_{h} \}\\
& = \tilde{\varphi}_{ij},
\end{align*}
where $\mathbb{E}_{F}(\cdot)$ is the fuzzy expected value operator. So the fuzzy expected value is just the fuzzification of the crisp expected value.\par
The hybrid model built over the log-Poisson regression can be defined in matrix form and by using result of section \ref{s1} as follows:
\begin{equation}\label{6}
\left\{ \begin{array}{cl}
\ln (Y^{c}) & = Y^{c^{\ast}} + \ln (\bm{\varepsilon}), \; Y^{c^{\ast}} = \mathbf{X}\bm{\beta} \\
\ln (Y^{L}) & = Y^{L^{\ast}} + \ln (\bm{\xi}), \; Y^{L^{\ast}} = Y^{c^{\ast}}\theta + \mathbf{1}\lambda\\
\ln (Y^{R}) & = Y^{R^{\ast}} + \ln (\bm{\eta}), \; Y^{R^{\ast}} = Y^{c^{\ast}}\delta + \mathbf{1}\mu
\end{array}\right.
\Leftrightarrow 
\left\{ \begin{array}{cl}
Y^{c^{\prime}} & = Y^{c^{\ast}} + \bm{\varepsilon}^{\prime}, \; Y^{c^{\ast}} = \mathbf{X}\bm{\beta} \\
Y^{L^{\prime}} & = Y^{L^{\ast}} + \bm{\xi}^{\prime}, \; Y^{L^{\ast}} = Y^{c^{\ast}}\theta + \mathbf{1}\lambda\\
Y^{R^{\prime}} & = Y^{R^{\ast}} + \bm{\eta}^{\prime}, \; Y^{R^{\ast}} = Y^{c^{\ast}}\delta + \mathbf{1}\mu
\end{array}\right.
\end{equation}
where
\begin{enumerate}
	\item[$\bullet$]
	$\bm{\beta} = \big(\tau, \bm{\alpha}, \bm{\gamma}\big)^{T} \in \mathbb{R}^{2k-1}$ with
	\begin{align*}
	            & \tau \in \mathbb{R}\\
	\bm{\alpha} & = (\alpha_{2} \; \ldots \; \alpha_{k}) \in \mathbb{R}^{k-1}\\
	\bm{\gamma} & = (\gamma_{2} \; \ldots \; \gamma_{k}) \in \mathbb{R}^{k-1}
	\end{align*}
	\item[$\bullet$]
	\begin{align*}
	Y^{c^{\prime}} & = \ln (Y^{c}); & Y^{L^{\prime}} & = \ln (Y^{L}); & Y^{R^{\prime}} & = \ln (Y^{R})\\
	\bm{\varepsilon}^{\prime} & = \ln(\bm{\varepsilon}); & \bm{\xi}^{\prime} & = \ln(\bm{\xi}); & \bm{\eta}^{\prime} & = \ln(\bm{\eta}).
	\end{align*}
	\item[$\bullet$]
$\bm{\varepsilon}, \bm{\xi}, \bm{\eta}$ are $n\times 1$ vectors of uncorrelated error terms following Poisson random variables ($\mathcal{P}(1)$) such that $\mathbb{E}(\bm{\varepsilon}^{\prime})=\mathbb{E}(\bm{\xi}^{\prime})=\mathbb{E}(\bm{\eta}^{\prime})=0_{n\times 1}$.
\end{enumerate}
\begin{theorem}\label{pro1}
	The iterative fuzzy least squares estimators $\hat{\bm{\beta}}, \hat{\theta}, \hat{\delta}, \hat{\lambda}$ and $\hat{\mu}$ of the parameters $\bm{\beta}, \theta, \delta, \lambda$ and $\mu$ in model \eqref{6} are given by :
	\begin{align}
	\hat{\bm{\beta}} & = \dfrac{1}{\big(1 + \theta^{2} + \delta^{2}\big)} (\mathbf{X}^{T} \mathbf{X})^{-1} \bigg\{\mathbf{X}^{T}Y^{c^{\prime}} + \theta\mathbf{X}^{T}\big(Y^{L^{\prime}} - \mathbf{1}\lambda \big) +  \delta\mathbf{X}^{T}\big(Y^{R^{\prime}} - \mathbf{1}\mu \big)\bigg\}\\
	\hat{\theta} & = \big[\bm{\beta}^{T}(\mathbf{X}^{T}\mathbf{X})\bm{\beta}\big]^{-1}\bm{\beta}^{T}\mathbf{X}^{T}\big(Y^{L^{\prime}} - \mathbf{1}\lambda \big)\\
	\hat{\delta} & = \big[\bm{\beta}^{T}(\mathbf{X}^{T}\mathbf{X})\bm{\beta}\big]^{-1}\bm{\beta}^{T}\mathbf{X}^{T}\big(Y^{R^{\prime}} - \mathbf{1}\mu \big)\\
	\hat{\lambda} & = \dfrac{1}{n}\bigg( \mathbf{1}^{T}Y^{L^{\prime}} - \bm{\beta}^{T}\mathbf{X}^{T}\mathbf{1}\theta\bigg)\\
	\hat{\mu} & = \dfrac{1}{n}\bigg( \mathbf{1}^{T}Y^{R^{\prime}} - \bm{\beta}^{T}\mathbf{X}^{T}\mathbf{1}\delta\bigg)
	\end{align}
	where $\bm{\beta}, \theta, \delta, \lambda$ and $\mu$ are the different values taken by the parameters before reaching their optimal values $\hat{\bm{\beta}}, \hat{\theta}, \hat{\delta}, \hat{\lambda}$ and $\hat{\mu}$.
\end{theorem}
\begin{proof}
By using the fuzzy least squares method on model \eqref{6}, we estimate its parameters as follows:\par
Denote
\begin{align*}
\mathcal{S}(\bm{\beta}, \theta, \delta, \lambda, \mu)   & =\big[\bm{\varepsilon}^{\prime}\big]^{T}\big[\bm{\varepsilon}^{\prime}\big] + \big[\bm{\xi}^{\prime}\big]^{T}\big[\bm{\xi}^{\prime}\big] + \big[\bm{\eta}^{\prime}\big]^{T}\big[\bm{\eta}^{\prime}\big]\\
\mathcal{S}(\bm{\beta}, \theta, \delta, \lambda, \mu) & = \bigg(Y^{c^{\prime}} - \mathbf{X}\bm{\beta}\bigg)^{T}\bigg(Y^{c^{\prime}} - \mathbf{X}\bm{\beta}\bigg) + \bigg( Y^{L^{\prime}} - \mathbf{X}\bm{\beta}\theta - \mathbf{1}\lambda\bigg)^{T}\bigg( Y^{L^{\prime}} - \mathbf{X}\bm{\beta}\theta - \mathbf{1}\lambda\bigg)\\
& + \bigg( Y^{R^{\prime}} - \mathbf{X}\bm{\beta}\delta - \mathbf{1}\mu\bigg)^{T}\bigg( Y^{R^{\prime}} - \mathbf{X}\bm{\beta}\delta - \mathbf{1}\mu\bigg)\\
\mathcal{S}(\bm{\beta}, \theta, \delta, \lambda, \mu) & = \bigg( \big[Y^{c^{\prime}}\big]^{T} - \bm{\beta}^{T}\mathbf{X}^{T}\bigg)\bigg(Y^{c^{\prime}} - \mathbf{X}\bm{\beta}\bigg) + \bigg(\big[Y^{L^{\prime}}\big]^{T} - \bm{\beta}^{T}\mathbf{X}^{T}\theta - \mathbf{1}^{T}\lambda\bigg)\bigg( Y^{L^{\prime}} - \mathbf{X}\bm{\beta}\theta - \mathbf{1}\lambda\bigg) \\
& + \bigg(\big[Y^{R^{\prime}}\big]^{T}- \bm{\beta}^{T}\mathbf{X}^{T}\delta - \mathbf{1}^{T}\mu\bigg)\bigg( Y^{R^{\prime}} - \mathbf{X}\bm{\beta}\delta - \mathbf{1}\mu\bigg)\\
\mathcal{S}(\bm{\beta}, \theta, \delta, \lambda, \mu) & = \big[Y^{c^{\prime}}\big]^{T}\big[Y^{c^{\prime}}\big] -2\big[Y^{c^{\prime}}\big]^{T}\mathbf{X}\bm{\beta} + \bm{\beta}^{T}(\mathbf{X}^{T}\mathbf{X})\bm{\beta} + \big[Y^{L^{\prime}}\big]^{T}\big[Y^{L^{\prime}}\big] - 2\big[Y^{L^{\prime}}\big]^{T}\mathbf{X}\bm{\beta}\theta \\
& - 2\big[Y^{L^{\prime}}\big]^{T}\mathbf{1}\lambda + \bm{\beta}^{T}(\mathbf{X}^{T}\mathbf{X})\bm{\beta}\theta^{2} + 2\mathbf{1}^{T}\mathbf{X}\bm{\beta}\theta\lambda + n\lambda^{2} + \big[Y^{R^{\prime}}\big]^{T}\big[Y^{R^{\prime}}\big] - 2\big[Y^{R^{\prime}}\big]^{T}\mathbf{X}\bm{\beta}\delta \\
& - 2\big[Y^{R^{\prime}}\big]^{T}\mathbf{1}\mu + \bm{\beta}^{T}(\mathbf{X}^{T}\mathbf{X})\bm{\beta}\delta^{2} + 2\mathbf{1}^{T}\mathbf{X}\bm{\beta}\delta\mu + n\mu^{2}
\end{align*}
\begin{align*}
\dfrac{\partial \mathcal{S}(\bm{\beta}, \theta, \delta, \lambda, \mu) }{\partial \bm{\beta}} & = - 2\big[Y^{c^{\prime}}\big]^{T}\mathbf{X} + 2(\mathbf{X}^{T}\mathbf{X})\bm{\beta} - 2\big[Y^{L^{\prime}}\big]^{T}\mathbf{X}\theta + 2(\mathbf{X}^{T}\mathbf{X})\bm{\beta}\theta^{2} + 2\mathbf{1}^{T}\mathbf{X}\theta\lambda - 2\big[Y^{R^{\prime}}\big]^{T}\mathbf{X}\delta \\
& + 2(\mathbf{X}^{T}\mathbf{X})\bm{\beta}\delta^{2} + 2\mathbf{1}^{T}\mathbf{X}\delta \mu\\
\dfrac{\partial \mathcal{S}(\bm{\beta}, \theta, \delta, \lambda, \mu) }{\partial \theta} & = -2\big[Y^{L^{\prime}}\big]^{T}\mathbf{X}\bm{\beta} + 2\bm{\beta}^{T}(\mathbf{X}^{T}\mathbf{X})\bm{\beta}\theta + 2\mathbf{1}^{T}\mathbf{X}\bm{\beta}\lambda \\
\dfrac{\partial \mathcal{S}(\bm{\beta}, \theta, \delta, \lambda, \mu) }{\partial \delta} & = -2\big[Y^{R^{\prime}}\big]^{T}\mathbf{X}\bm{\beta} + 2\bm{\beta}^{T}(\mathbf{X}^{T}\mathbf{X})\bm{\beta}\delta + 2\mathbf{1}^{T}\mathbf{X}\bm{\beta}\mu \\
\dfrac{\partial \mathcal{S}(\bm{\beta}, \theta, \delta, \lambda, \mu) }{\partial \lambda} & = -2\big[Y^{L^{\prime}}\big]^{T}\mathbf{1} + 2\mathbf{1}^{T}\mathbf{X}\bm{\beta}\theta + 2n\lambda\\
\dfrac{\partial \mathcal{S}(\bm{\beta}, \theta, \delta, \lambda, \mu) }{\partial \mu} & = -2\big[Y^{R^{\prime}}\big]^{T}\mathbf{1} + 2\mathbf{1}^{T}\mathbf{X}\bm{\beta}\delta + 2n\mu.
\end{align*}
\begin{align}\label{14}
\dfrac{\partial \mathcal{S}(\bm{\beta}, \theta, \delta, \lambda, \mu) }{\partial \bm{\beta}}  = 0 & \Leftrightarrow  -\big[Y^{c^{\prime}}\big]^{T}\mathbf{X} + (\mathbf{X}^{T}\mathbf{X})\bm{\beta}(1 + \theta^{2} + \delta^{2}) - \big[Y^{L^{\prime}}\big]^{T}\mathbf{X}\theta- \big[Y^{R^{\prime}}\big]^{T}\mathbf{X}\delta + \mathbf{1}^{T}\mathbf{X}\theta\lambda + \mathbf{1}^{T}\mathbf{X}\delta\mu = 0 \\
& \Leftrightarrow \bm{\beta} = \dfrac{1}{\big(1 + \theta^{2} + \delta^{2}\big)} (\mathbf{X}^{T} \mathbf{X})^{-1} \bigg\{\big[Y^{c^{\prime}}\big]^{T}\mathbf{X} + \big(\big[Y^{L^{\prime}}\big]^{T} - \mathbf{1}^{T}\lambda \big)\mathbf{X}\theta +  \big(\big[Y^{R^{\prime}}\big]^{T} - \mathbf{1}^{T}\mu \big)\mathbf{X}\delta\bigg\}\\
\dfrac{\partial \mathcal{S}(\bm{\beta}, \theta, \delta, \lambda, \mu) }{\partial \theta} = 0 & \Leftrightarrow -\big[Y^{L^{\prime}}\big]^{T}\mathbf{X}\bm{\beta} + \bm{\beta}^{T}(\mathbf{X}^{T}\mathbf{X})\bm{\beta}\theta + \mathbf{1}^{T}\mathbf{X}\bm{\beta}\lambda = 0\\
& \Leftrightarrow \theta = \big[\bm{\beta}^{T}(\mathbf{X}^{T}\mathbf{X})\bm{\beta}\big]^{-1}\big(\big[Y^{L^{\prime}}\big]^{T} - \mathbf{1}^{T}\lambda \big)\mathbf{X}\bm{\beta}\\
\dfrac{\partial \mathcal{S}(\bm{\beta}, \theta, \delta, \lambda, \mu) }{\partial \delta} = 0 & \Leftrightarrow -\big[Y^{R^{\prime}}\big]^{T}\mathbf{X}\bm{\beta} + \bm{\beta}^{T}(\mathbf{X}^{T}\mathbf{X})\bm{\beta}\delta + \mathbf{1}^{T}\mathbf{X}\bm{\beta}\mu = 0\\
& \Leftrightarrow \delta = \big[\bm{\beta}^{T}(\mathbf{X}^{T}\mathbf{X})\bm{\beta}\big]^{-1}\big(\big[Y^{R^{\prime}}\big]^{T} - \mathbf{1}^{T}\mu \big)\mathbf{X}\bm{\beta}\\
\dfrac{\partial \mathcal{S}(\bm{\beta}, \theta, \delta, \lambda, \mu) }{\partial \lambda}  = 0 & \Leftrightarrow -\big[Y^{L^{\prime}}\big]^{T}\mathbf{1} + \mathbf{1}^{T}\mathbf{X}\bm{\beta}\theta + n\lambda = 0\\
& \Leftrightarrow \lambda = \dfrac{1}{n}\bigg( \big[Y^{L^{\prime}}\big]^{T} \mathbf{1} - \mathbf{1}^{T}\mathbf{X}\bm{\beta}\theta\bigg)\\
\dfrac{\partial \mathcal{S}(\bm{\beta}, \theta, \delta, \lambda, \mu) }{\partial \mu}  = 0 & \Leftrightarrow -\big[Y^{R^{\prime}}\big]^{T}\mathbf{1} + \mathbf{1}^{T}\mathbf{X}\bm{\beta}\delta + n\mu = 0\\
& \Leftrightarrow \mu = \dfrac{1}{n}\bigg( \big[Y^{R^{\prime}}\big]^{T} \mathbf{1} - \mathbf{1}^{T}\mathbf{X}\bm{\beta}\delta\bigg)
\end{align}
Furthermore
\begin{align}
\dfrac{\partial^{2} \mathcal{S}(\bm{\beta}, \theta, \delta, \lambda, \mu) }{\partial \bm{\beta}^2} & = 2(\mathbf{X}^{T}\mathbf{X})(1 + \theta^{2} + \delta^{2})\\
\dfrac{\partial^{2} \mathcal{S}(\bm{\beta}, \theta, \delta, \lambda, \mu) }{\partial \theta\partial \bm{\beta}} & = -2\big[Y^{L^{\prime}}\big]^{T}\mathbf{X} + 4(\mathbf{X}^{T}\mathbf{X})\bm{\beta}\theta + 2\mathbf{1}^{T}\mathbf{X}\lambda\\
\dfrac{\partial^{2} \mathcal{S}(\bm{\beta}, \theta, \delta, \lambda, \mu) }{\partial \delta\partial \bm{\beta}} & = -2\big[Y^{R^{\prime}}\big]^{T}\mathbf{X} + 4(\mathbf{X}^{T}\mathbf{X})\bm{\beta}\delta + 2\mathbf{1}^{T}\mathbf{X}\mu\\
\dfrac{\partial^{2} \mathcal{S}(\bm{\beta}, \theta, \delta, \lambda, \mu) }{\partial \lambda\partial \bm{\beta}} & =2\mathbf{1}^{T}\mathbf{X}\theta, \quad \dfrac{\partial^{2} \mathcal{S}(\bm{\beta}, \theta, \delta, \lambda, \mu) }{\partial \mu\partial \bm{\beta}} =2\mathbf{1}^{T}\mathbf{X}\delta\\
\dfrac{\partial^{2} \mathcal{S}(\bm{\beta}, \theta, \delta, \lambda, \mu) }{\partial \bm{\beta}\partial \theta} & = -2\big[Y^{L^{\prime}}\big]^{T}\mathbf{X} + 4(\mathbf{X}^{T}\mathbf{X})\bm{\beta}\theta + 2\mathbf{1}^{T}\mathbf{X}\lambda\\
\dfrac{\partial^{2} \mathcal{S}(\bm{\beta}, \theta, \delta, \lambda, \mu) }{\partial \theta^2} & =2\bm{\beta}^{T}(\mathbf{X}^{T}\mathbf{X})\bm{\beta}, \quad \dfrac{\partial^{2} \mathcal{S}(\bm{\beta}, \theta, \delta, \lambda, \mu) }{\partial \delta\partial \theta} = 0\\
\dfrac{\partial^{2} \mathcal{S}(\bm{\beta}, \theta, \delta, \lambda, \mu) }{\partial \lambda\partial \theta} & = 2\mathbf{1}^{T}\mathbf{X}\bm{\beta}, \quad \dfrac{\partial^{2} \mathcal{S}(\bm{\beta}, \theta, \delta, \lambda, \mu) }{\partial \mu\partial \theta}  = 0\\
\dfrac{\partial^{2} \mathcal{S}(\bm{\beta}, \theta, \delta, \lambda, \mu) }{\partial \bm{\beta}\partial \delta} & = -2\big[Y^{R^{\prime}}\big]^{T}\mathbf{X} + 4(\mathbf{X}^{T}\mathbf{X})\bm{\beta}\delta + 2\mathbf{1}^{T}\mathbf{X}\mu\\
\dfrac{\partial^{2} \mathcal{S}(\bm{\beta}, \theta, \delta, \lambda, \mu) }{\partial \theta \partial \delta} & = 0, \quad \dfrac{\partial^{2} \mathcal{S}(\bm{\beta}, \theta, \delta, \lambda, \mu) }{\partial \delta^2} =2\bm{\beta}^{T}(\mathbf{X}^{T}\mathbf{X})\bm{\beta}\\
\dfrac{\partial^{2} \mathcal{S}(\bm{\beta}, \theta, \delta, \lambda, \mu) }{\partial \lambda\partial \delta} & = 0, \quad \dfrac{\partial^{2} \mathcal{S}(\bm{\beta}, \theta, \delta, \lambda, \mu) }{\partial \mu\partial \delta} = 2\mathbf{1}^{T}\mathbf{X}\bm{\beta}\\
\dfrac{\partial^{2} \mathcal{S}(\bm{\beta}, \theta, \delta, \lambda, \mu) }{\partial \bm{\beta}\partial \lambda} & = 2\mathbf{1}^{T}\mathbf{X}\theta, \quad \dfrac{\partial^{2} \mathcal{S}(\bm{\beta}, \theta, \delta, \lambda, \mu) }{\partial \theta\partial \lambda} = 2\mathbf{1}^{T}\mathbf{X}\bm{\beta}\\
\dfrac{\partial^{2} \mathcal{S}(\bm{\beta}, \theta, \delta, \lambda, \mu) }{\partial \delta\partial \lambda} & = 0, \quad\dfrac{\partial^{2} \mathcal{S}(\bm{\beta}, \theta, \delta, \lambda, \mu) }{\partial \lambda^2} = 2n\\
\dfrac{\partial^{2} \mathcal{S}(\bm{\beta}, \theta, \delta, \lambda, \mu) }{\partial \mu\partial \lambda} & = 0, \quad\dfrac{\partial^{2} \mathcal{S}(\bm{\beta}, \theta, \delta, \lambda, \mu) }{\partial \bm{\beta}\partial \mu} = 2\mathbf{1}^{T}\mathbf{X}\delta\\
\dfrac{\partial^{2} \mathcal{S}(\bm{\beta}, \theta, \delta, \lambda, \mu) }{\partial \theta\partial \mu} & = 0, \quad\dfrac{\partial^{2} \mathcal{S}(\bm{\beta}, \theta, \delta, \lambda, \mu) }{\partial \delta\partial \mu} = 2\mathbf{1}^{T}\mathbf{X}\bm{\beta}\\
\dfrac{\partial^{2} \mathcal{S}(\bm{\beta}, \theta, \delta, \lambda, \mu) }{\partial \lambda\partial \mu} & = 0, \quad\dfrac{\partial^{2} \mathcal{S}(\bm{\beta}, \theta, \delta, \lambda, \mu) }{\partial \mu^2} = 2n.
\end{align}
Let $H$ be the corresponding Hessian matrix. Then 
\begin{equation}
\forall \mathbf{u}=(u_{1}\; \ldots \; u_{5})^{T}\in \mathbb{R}^{5}, \quad \mathbf{u}^{T}H\mathbf{u}= 2(\mathbf{X}^{T}\mathbf{X})(1 + \theta^{2})u_{1}^{2} + 2\bm{\beta}^{T}(\mathbf{X}^{T}\mathbf{X})\bm{\beta}(u_{1}^{2} + u_{3}^{2}) + 2n(u_{4}^{2} + u_{5}^{2})> 0
\end{equation}
Hence $H$ is a semi positive definite matrix and this ends the proof.
\end{proof}
\section{A Goodness of Fit Index for Hybrid log-Poisson Regression}
In this section, we derive a goodness of fit index $\tilde{R}^{2}_{F}$ for the hybrid log-Poisson regression. This index is relevant to assess the explanatory power of the model.\par
In order to provide the mathematical formula of that $\tilde{R}^{2}_{F}$, let us prove some results.
\begin{proposition}\label{prop}
Let 
\begin{equation}
\left\{ \begin{array}{cl}
Y^{c^{\prime}} & = Y^{c^{\ast}} + \bm{\varepsilon}^{\prime}, \; Y^{c^{\ast}} = \mathbf{X}\bm{\beta} \\
Y^{L^{\prime}} & = Y^{L^{\ast}} + \bm{\xi}^{\prime}, \; Y^{L^{\ast}} = Y^{c^{\ast}}\theta + \mathbf{1}\lambda\\
Y^{R^{\prime}} & = Y^{R^{\ast}} + \bm{\eta}^{\prime}, \; Y^{R^{\ast}} = Y^{c^{\ast}}\delta + \mathbf{1}\mu
\end{array}\right.
\end{equation}
be the hybrid model built over the classical log-Poisson regression,
where
\begin{enumerate}
	\item[$\bullet$]
	$\bm{\beta} = \big(\tau, \bm{\alpha}, \bm{\gamma}\big)^{T} \in \mathbb{R}^{2k-1}$ with
	\begin{align*}
	& \tau \in \mathbb{R}\\
	\bm{\alpha} & = (\alpha_{2} \; \ldots \; \alpha_{k}) \in \mathbb{R}^{k-1}\\
	\bm{\gamma} & = (\gamma_{2} \; \ldots \; \gamma_{k}) \in \mathbb{R}^{k-1}
	\end{align*}
	\item[$\bullet$]
	$\bm{\varepsilon}, \bm{\xi}, \bm{\eta}$ are $n\times 1$ vectors of error terms following Poisson random variables ($\mathcal{P}(1)$).
\end{enumerate} 
 The following relationships hold :
\begin{enumerate}
\item[\textbf{1)}]$\sum\limits_{i=1}^{k} \sum\limits_{j=1}^{k-i+1}\big(Y_{ij}^{c^{\ast}}\big)^{T}\big(Y_{ij}^{c^{\prime}} - Y_{ij}^{c^{\ast}}\big)=0$
\item[\textbf{2)}]
\begin{itemize}
\item[$\bullet$]$\sum\limits_{i=1}^{k} \sum\limits_{j=1}^{k-i+1}\big( Y_{ij}^{c^{\prime}} - Y_{ij}^{c^{\ast}}\big) = 0$
\item[$\bullet$]$\sum\limits_{i=1}^{k} \sum\limits_{j=1}^{k-i+1}\big( Y_{ij}^{L^{\prime}} - Y_{ij}^{L^{\ast}}\big) = 0$
\item[$\bullet$]$\sum\limits_{i=1}^{k} \sum\limits_{j=1}^{k-i+1}\big( Y_{ij}^{R^{\prime}} - Y_{ij}^{R^{\ast}}\big) = 0$
\end{itemize}
\item[\textbf{3)}]
\begin{itemize}
\item[$\bullet$]$\sum\limits_{i=1}^{k} \sum\limits_{j=1}^{k-i+1}\big( Y_{ij}^{L^{\prime}} - Y_{ij}^{L^{\ast}}\big)^{T}Y_{ij}^{L^{\ast}} = 0$
\item[$\bullet$]$\sum\limits_{i=1}^{k} \sum\limits_{j=1}^{k-i+1}\big( Y_{ij}^{R^{\prime}} - Y_{ij}^{R^{\ast}}\big)^{T}Y_{ij}^{R^{\ast}} = 0$
\end{itemize}
\end{enumerate}
\end{proposition}

\begin{proof}
\begin{enumerate}
\item[\textbf{1)}]
We have\\
$\mathcal{S}(\bm{\beta}, \theta, \delta, \lambda,\mu)  = \bigg\{\big[Y^{c^{\prime}}\big]^{T}\big[Y^{c^{\prime}}\big] - 2\big[Y^{c^{\prime}}\big]^{T}\mathbf{X}\bm{\beta} + \bm{\beta}^{T}(\mathbf{X}^{T}\mathbf{X})\bm{\beta} \bigg\} 
+ \bigg\{\big[Y^{L^{\prime}}\big]^{T}\big[Y^{L^{\prime}}\big] - 2\big[Y^{L^{\prime}}\big]^{T}\big( \mathbf{X}\bm{\beta}\theta + \mathbf{1}\lambda\big) \bigg\} + 
\bigg\{\big[Y^{R^{\prime}}\big]^{T}\big[Y^{R^{\prime}}\big] - 2\big[Y^{R^{\prime}}\big]^{T}\big( \mathbf{X}\bm{\beta}\delta + \mathbf{1}\mu\big) \bigg\} 
+ \big( \mathbf{X}\bm{\beta}\theta + \mathbf{1}\lambda\big)^{T}\big( \mathbf{X}\bm{\beta}\theta + \mathbf{1}\lambda\big) + \big( \mathbf{X}\bm{\beta}\delta + \mathbf{1}\mu\big)^{T}\big( \mathbf{X}\bm{\beta}\delta + \mathbf{1}\mu\big)$
\begin{align}\label{7}\nonumber
\dfrac{\partial \mathcal{S}(\bm{\beta}, \theta, \delta, \lambda,\mu)}{\partial \bm{\beta}} = 0_{2k-1} & \Leftrightarrow \mathbf{X}^{T}\bigg\{\big(Y^{c^{\prime}} - Y^{c^{\ast}} \big) + \theta \big(Y^{L^{\prime}} - Y^{c^{\ast}}\theta - \mathbf{1}\lambda \big) + \delta \big(Y^{R^{\prime}} - Y^{c^{\ast}}\delta - \mathbf{1}\mu \big)\bigg\}=0_{2k-1}\\ 
& \Leftrightarrow \mathbf{X}^{T}\bigg\{\big(Y^{c^{\prime}} - Y^{c^{\ast}} \big) + \theta \big(Y^{L^{\prime}} - Y^{L^{\ast}} \big) + \delta \big(Y^{R^{\prime}} - Y^{R^{\ast}} \big)\bigg\}=0_{2k-1}
\end{align}
\begin{align}\label{8}\nonumber
\dfrac{\partial \mathcal{S}(\bm{\beta}, \theta, \delta, \lambda,\mu)}{\partial \theta} = 0 & \Leftrightarrow \bm{\beta}^{T}(\mathbf{X}^{T}\mathbf{X})\bm{\beta}\theta - \bm{\beta}^{T}\mathbf{X}^{T}\big[Y^{L^{\prime}}\big] + \bm{\beta}^{T}\mathbf{X}^{T} \mathbf{1}\lambda = 0\\\nonumber
& \Leftrightarrow \bm{\beta}^{T}(\mathbf{X}^{T}\mathbf{X})\bm{\beta}\theta = \bm{\beta}^{T}\mathbf{X}^{T}\big[Y^{L^{\prime}}\big] - \bm{\beta}^{T}\mathbf{X}^{T} \mathbf{1}\lambda \\\nonumber
& \Leftrightarrow \big(Y^{c^{\ast}}\big)^{T}\big(Y^{c^{\ast}}\big)\theta = \big(Y^{c^{\ast}}\big)^{T} \bigg(Y^{L^{\prime}} - \mathbf{1}\lambda - Y^{c^{\ast}}\theta + Y^{c^{\ast}}\theta \bigg)\\\nonumber
& \Leftrightarrow \big(Y^{c^{\ast}}\big)^{T}\big(Y^{c^{\ast}}\big)\theta = \big(Y^{c^{\ast}}\big)^{T} \bigg(Y^{L^{\prime}} - Y^{L^{\ast}} + Y^{c^{\ast}}\theta \bigg)\\\nonumber
& \Leftrightarrow \big(Y^{c^{\ast}}\big)^{T}\big(Y^{c^{\ast}}\big)\theta = \big(Y^{c^{\ast}}\big)^{T} \big(Y^{L^{\prime}} - Y^{L^{\ast}}\big) + \big(Y^{c^{\ast}}\big)^{T}\big(Y^{c^{\ast}}\big)\theta\\
& \Leftrightarrow \big(Y^{c^{\ast}}\big)^{T} \big(Y^{L^{\prime}} - Y^{L^{\ast}}\big)=0 
\end{align}
\begin{align}\label{9}\nonumber
\dfrac{\partial \mathcal{S}(\bm{\beta}, \theta, \delta, \lambda,\mu)}{\partial \delta} = 0 & \Leftrightarrow -\big[Y^{R^{\prime}}\big]^{T}\mathbf{X}\bm{\beta} + \bm{\beta}(\mathbf{X}^{T}\mathbf{X})\bm{\beta}\delta + \mathbf{1}^{T}\mathbf{X}\bm{\beta}\mu = 0 \\\nonumber
& \Leftrightarrow -\bm{\beta}^{T}\mathbf{X}^{T}\big[Y^{R^{\prime}}\big] + \bm{\beta}(\mathbf{X}^{T}\mathbf{X})\bm{\beta}\delta + \bm{\beta}^{T}\mathbf{X}^{T}\mathbf{1}\mu = 0 \\\nonumber
& \Leftrightarrow \bm{\beta}(\mathbf{X}^{T}\mathbf{X})\bm{\beta}\delta = \bm{\beta}^{T}\mathbf{X}^{T}\big[Y^{R^{\prime}}\big] - \bm{\beta}^{T}\mathbf{X}^{T}\mathbf{1}\mu \\\nonumber
& \Leftrightarrow \big(Y^{c^{\ast}}\big)^{T}\big(Y^{c^{\ast}}\big)\delta = \big(Y^{c^{\ast}}\big)^{T} \bigg\{Y^{R^{\prime}} - \mathbf{1}\mu \bigg\}\\\nonumber
&  \Leftrightarrow \big(Y^{c^{\ast}}\big)^{T}\big(Y^{c^{\ast}}\big)\delta = \big(Y^{c^{\ast}}\big)^{T} \bigg\{Y^{R^{\prime}} - Y^{c^{\ast}}\delta - \mathbf{1}\mu + Y^{c^{\ast}}\delta\bigg\}\\\nonumber
&  \Leftrightarrow \big(Y^{c^{\ast}}\big)^{T}\big(Y^{c^{\ast}}\big)\delta = \big(Y^{c^{\ast}}\big)^{T} \bigg\{Y^{R^{\prime}} - Y^{R^{\ast}} + Y^{c^{\ast}}\delta\bigg\}\\\nonumber
&  \Leftrightarrow \big(Y^{c^{\ast}}\big)^{T}\big(Y^{c^{\ast}}\big)\delta = \big(Y^{c^{\ast}}\big)^{T} \bigg\{Y^{R^{\prime}} - Y^{R^{\ast}}\bigg\}  + \big(Y^{c^{\ast}}\big)^{T}\big(Y^{c^{\ast}}\big)\delta\\
& \Leftrightarrow \big(Y^{c^{\ast}}\big)^{T} \bigg\{Y^{R^{\prime}} - Y^{R^{\ast}}\bigg\} = 0
\end{align}
\begin{align}\label{10}\nonumber
\eqref{7} & \Leftrightarrow \bm{\beta}^{T}\mathbf{X}^{T}\bigg\{\big(Y^{c^{\prime}} - Y^{c^{\ast}} \big) + \theta \big(Y^{L^{\prime}} - Y^{L^{\ast}} \big) + \delta \big(Y^{R^{\prime}} - Y^{R^{\ast}} \big)\bigg\}=0\\
& \Leftrightarrow \big(Y^{c^{\ast}}\big)^{T}\bigg\{\big(Y^{c^{\prime}} - Y^{c^{\ast}} \big) + \theta \big(Y^{L^{\prime}} - Y^{L^{\ast}} \big) + \delta \big(Y^{R^{\prime}} - Y^{R^{\ast}} \big)\bigg\}=0
\end{align}
\begin{equation}\label{11}
\eqref{8}-\eqref{9}\; \text{in}\; \eqref{10} \Rightarrow \big(Y^{c^{\ast}}\big)^{T}\big(Y^{c^{\prime}} - Y^{c^{\ast}}\big)=0.\quad \square
\end{equation}
\item[\textbf{2)}]
\begin{align}
\mathcal{S}(\bm{\beta}, \theta, \delta, \lambda,\mu)  & = \begin{Vmatrix}
\big( Y_{ij}^{c^{\prime}} - Y_{ij}^{c^{\ast}}\big)\\
\big( Y_{ij}^{L^{\prime}} - Y_{ij}^{L^{\ast}}\big)\\
\big( Y_{ij}^{R^{\prime}} - Y_{ij}^{R^{\ast}}\big)
\end{Vmatrix}^2\\
& =  \begin{Vmatrix}
\big( Y_{ij}^{c^{\prime}} - \mathbf{X}\bm{\beta}\big)\\
\big( Y_{ij}^{L^{\prime}} - \mathbf{X}\bm{\beta}\theta - \mathbf{1}\lambda\big)\\
\big( Y_{ij}^{R^{\prime}} - \mathbf{X}\bm{\beta}\delta - \mathbf{1}\mu\big)
\end{Vmatrix}^2\\
& = \begin{Vmatrix}
\begin{pmatrix}
Y^{c^{\prime}}\\
Y^{L^{\prime}} - \mathbf{1}\lambda \\
Y^{R^{\prime}} - \mathbf{1}\mu 
\end{pmatrix}
-
\begin{pmatrix}
\mathbf{X}\\
\mathbf{X}\theta \\
\mathbf{X}\delta
\end{pmatrix}
\bm{\beta}
\end{Vmatrix}^2
\end{align}
\begin{align*}
\min_{\bm{\beta}}\mathcal{S}(\bm{\beta}, \theta, \delta, \lambda,\mu) & \Leftrightarrow \mathbf{1}^{T}\bigg( Y^{c^{\prime}} + Y^{L^{\prime}} - \mathbf{1}\lambda + Y^{R^{\prime}} - \mathbf{1}\mu\bigg) = \mathbf{1}^{T}\begin{pmatrix}
\mathbf{X}\\
\mathbf{X}\theta\\
\mathbf{X}\delta
\end{pmatrix}
\bm{\beta}\\
& \Leftrightarrow \mathbf{1}^{T}\big[Y^{c^{\prime}}\big] + \mathbf{1}^{T}\big[Y^{L^{\prime}}\big] - n\lambda + \mathbf{1}^{T}\big[Y^{R^{\prime}}\big] - n\mu = \mathbf{1}^{T}\mathbf{X}\bm{\beta} + \mathbf{1}^{T}\mathbf{X}\bm{\beta}\theta + \mathbf{1}^{T}\mathbf{X}\bm{\beta}\delta \\
& \Leftrightarrow \mathbf{1}^{T}\big[Y^{c^{\prime}}\big] + \mathbf{1}^{T}\big[Y^{L^{\prime}}\big] + \mathbf{1}^{T}\big[Y^{R^{\prime}}\big] - n\lambda - n\mu - \mathbf{1}^{T}\mathbf{X}\bm{\beta} - \mathbf{1}^{T}\mathbf{X}\bm{\beta}\theta - \mathbf{1}^{T}\mathbf{X}\bm{\beta}\delta =0\\
& \Leftrightarrow \mathbf{1}^{T}\big(\big[Y^{c^{\prime}}\big] - \mathbf{X}\bm{\beta}\big) + \mathbf{1}^{T}\big(\big[Y^{L^{\prime}}\big] - \mathbf{X}\bm{\beta}\theta - \mathbf{1}\lambda\big) + \mathbf{1}^{T}\big(\big[Y^{L^{\prime}}\big] - \mathbf{X}\bm{\beta}\delta - \mathbf{1}\mu\big) = 0
\end{align*}
\begin{equation}\label{12}
\Leftrightarrow \mathbf{1}^{T}\big(Y^{c^{\prime}} - \mathbf{X}\bm{\beta}\big) = - \mathbf{1}^{T}\big(Y^{L^{\prime}} - \mathbf{X}\bm{\beta}\theta - \mathbf{1}\lambda\big) - \mathbf{1}^{T}\big(Y^{R^{\prime}} - \mathbf{X}\bm{\beta}\delta - \mathbf{1}\mu\big)
\end{equation}
\begin{equation}\label{13}
\dfrac{\partial \mathcal{S}(\bm{\beta}, \theta, \delta, \lambda, \mu) }{\partial \theta} = 0 \Leftrightarrow - (Y^{c^{\ast}})^{T}\big[Y^{L^{\prime}}\big] + (Y^{c^{\ast}})^{T}(Y^{c^{\ast}})\theta + (Y^{c^{\ast}})^{T}\mathbf{1}\lambda = 0
\end{equation}
\begin{equation}\label{15}
\dfrac{\partial \mathcal{S}(\bm{\beta}, \theta, \delta, \lambda, \mu) }{\partial \delta} = 0 \Leftrightarrow - (Y^{c^{\ast}})^{T}\big[Y^{R^{\prime}}\big] + (Y^{c^{\ast}})^{T}(Y^{c^{\ast}})\delta + (Y^{c^{\ast}})^{T}\mathbf{1}\mu = 0
\end{equation}
\begin{equation}
\eqref{13}\Leftrightarrow (Y^{c^{\ast}})^{T} \bigg\{Y^{L^{\prime}} - Y^{c^{\ast}}\theta - \mathbf{1}\lambda\bigg\}=0
\end{equation}
\begin{equation}\label{16}
\Leftrightarrow (Y^{c^{\ast}})^{T}\big( Y^{L^{\prime}} - Y^{L^{\ast}}\big)=0
\end{equation}
\begin{equation}\label{17}
\Leftrightarrow \mathbf{1}^{T}\big( Y^{L^{\prime}} - Y^{L^{\ast}}\big)=0
\end{equation}
\begin{equation}
\eqref{15}\Leftrightarrow (Y^{c^{\ast}})^{T} \bigg\{Y^{R^{\prime}} - Y^{c^{\ast}}\delta - \mathbf{1}\mu\bigg\}=0
\end{equation}
\begin{equation}\label{18}
\Leftrightarrow (Y^{c^{\ast}})^{T} \bigg (Y^{R^{\prime}} - Y^{R^{\ast}}\bigg)=0
\end{equation}
\begin{equation}\label{19}
\Leftrightarrow \mathbf{1}^{T} \bigg (Y^{R^{\prime}} - Y^{R^{\ast}}\bigg)=0
\end{equation}
\begin{equation}\label{20}
\eqref{12} \Leftrightarrow \mathbf{1}^{T}\big(Y^{c^{\prime}} - Y^{c^{\ast}}\big) = - \mathbf{1}^{T}\big(Y^{L^{\prime}} - Y^{L^{\ast}}\big) - \mathbf{1}^{T}\big(Y^{R^{\prime}} - Y^{R^{\ast}}\big)
\end{equation}
\begin{equation*}
\eqref{17}-\eqref{19}\; \text{in}\; \eqref{20}\Rightarrow \mathbf{1}^{T}\big(Y^{c^{\prime}} - Y^{c^{\ast}}\big) \quad \square .
\end{equation*}
\item[\textbf{3)}]
\begin{align}\label{21}\nonumber
\bigg (Y^{L^{\prime}} - Y^{L^{\ast}}\bigg)^{T}(Y^{L^{\ast}}) & = \bigg (Y^{L^{\prime}} - Y^{L^{\ast}}\bigg)^{T}\bigg(Y^{c^{\ast}}\theta + \mathbf{1}\lambda \bigg)\\ 
& = \bigg (Y^{L^{\prime}} - Y^{L^{\ast}}\bigg)^{T}Y^{c^{\ast}}\theta + \bigg (Y^{L^{\prime}} - Y^{L^{\ast}}\bigg)^{T}\mathbf{1}\lambda 
\end{align}
\begin{align}
& \eqref{16}\Rightarrow  \bigg (Y^{L^{\prime}} - Y^{L^{\ast}}\bigg)^{T}Y^{c^{\ast}}\theta = 0\\
& \eqref{17}\Rightarrow  \bigg (Y^{L^{\prime}} - Y^{L^{\ast}}\bigg)^{T}\mathbf{1}\lambda = 0\\
& \eqref{16}-\eqref{17} \; \text{in}\; \eqref{21} \Rightarrow \bigg (Y^{L^{\prime}} - Y^{L^{\ast}}\bigg)^{T}(Y^{L^{\ast}}) = 0
\end{align}
Similarly
\begin{align}\label{22}\nonumber
\bigg (Y^{R^{\prime}} - Y^{R^{\ast}}\bigg)^{T}(Y^{R^{\ast}}) & = \bigg (Y^{R^{\prime}} - Y^{R^{\ast}}\bigg)^{T}\bigg(Y^{c^{\ast}}\delta + \mathbf{1}\mu \bigg)\\ 
& = \bigg (Y^{R^{\prime}} - Y^{R^{\ast}}\bigg)^{T}Y^{c^{\ast}}\delta + \bigg (Y^{R^{\prime}} - Y^{R^{\ast}}\bigg)^{T}\mathbf{1}\mu 
\end{align}
\begin{equation}
\eqref{18}-\eqref{19}\; \text{in}\; \eqref{22}\Rightarrow \bigg (Y^{R^{\prime}} - Y^{R^{\ast}}\bigg)^{T}(Y^{R^{\ast}}) \quad\square .
\end{equation}
\end{enumerate}
\end{proof}
\begin{definition}
For a set of crisp observations $\mathbf{X}_{n\times (p+1)}$ and by considering the hybrid log-Poisson model built over the classical one \eqref{6} in loss reserving framework.\par
We define for the fuzzy output $\tilde{Y}_{ij}=(Y_{ij}^{L}, Y_{ij}^{c}, Y_{ij}^{R}), \; (i,j)\in \{1, \ldots, k\} \times \{ 1, \ldots , k-i+1\}$ the following concepts :
\begin{enumerate}
\item[\textbf{1)}] The fuzzy total sum of squares
\begin{equation}
FSST=\sum\limits_{i=1}^{k}\sum\limits_{j=1}^{k-i+1}\big(Y_{ij}^{c^{\prime}} - \overline{Y}_{\ln}^{c}\big)^{2} + \sum\limits_{i=1}^{k}\sum\limits_{j=1}^{k-i+1}\big(Y_{ij}^{L^{\prime}} - \overline{Y}_{\ln}^{L}\big)^{2} + \sum\limits_{i=1}^{k}\sum\limits_{j=1}^{k-i+1}\big(Y_{ij}^{R^{\prime}} - \overline{Y}_{\ln}^{R}\big)^{2}
\end{equation}
where
\begin{align}
\overline{Y}_{\ln}^{c} & = \dfrac{1}{n}\sum\limits_{i=1}^{k}\sum\limits_{j=1}^{k-i+1}Y_{ij}^{c^{\prime}}\\
\overline{Y}_{\ln}^{L} & = \dfrac{1}{n}\sum\limits_{i=1}^{k}\sum\limits_{j=1}^{k-i+1}Y_{ij}^{L^{\prime}}\\
\overline{Y}_{\ln}^{R} & = \dfrac{1}{n}\sum\limits_{i=1}^{k}\sum\limits_{j=1}^{k-i+1}Y_{ij}^{R^{\prime}}\\
                     n & = \dfrac{1}{2}k(k+1)
\end{align}
\item[\textbf{2)}] The fuzzy sum of the squares of the regression
\begin{equation}
FSSR=\sum\limits_{i=1}^{k}\sum\limits_{j=1}^{k-i+1}\big(Y_{ij}^{c^{\ast}} - \overline{Y}_{\ln}^{c}\big)^{2} + \sum\limits_{i=1}^{k}\sum\limits_{j=1}^{k-i+1}\big(Y_{ij}^{L^{\ast}} - \overline{Y}_{\ln}^{L}\big)^{2} + \sum\limits_{i=1}^{k}\sum\limits_{j=1}^{k-i+1}\big(Y_{ij}^{R^{\ast}} - \overline{Y}_{\ln}^{R}\big)^{2}
\end{equation}
\item[\textbf{3)}] The fuzzy sum of the squares of errors
\begin{equation}
FSSE=\sum\limits_{i=1}^{k}\sum\limits_{j=1}^{k-i+1}\big(Y_{ij}^{c^{\prime}} - Y_{ij}^{c^{\ast}}\big)^{2} + \sum\limits_{i=1}^{k}\sum\limits_{j=1}^{k-i+1}\big(Y_{ij}^{L^{\prime}} - Y_{ij}^{L^{\ast}}\big)^{2} + \sum\limits_{i=1}^{k}\sum\limits_{j=1}^{k-i+1}\big(Y_{ij}^{R^{\prime}} - Y_{ij}^{R^{\ast}}\big)^{2}
\end{equation}
\end{enumerate}
\end{definition}
\begin{theorem}\label{prop2}
	Let us consider a set of crisp observations $\mathbf{X}_{n\times (p+1)}$ and fuzzy output $\tilde{Y}_{ij}=(Y_{ij}^{L}, Y_{ij}^{c}, Y_{ij}^{R}), \; (i,j)\in \{1, \ldots, k\} \times \{ 1, \ldots , k-i+1\}$. By considering the hybrid log-Poisson model built over the classical one \eqref{6} in loss reserving framework, the following relationship holds :
	\begin{equation}
	FSST = FSSR + FSSE
	\end{equation} 
\end{theorem}
\begin{proof}
\begin{equation}\label{23}
FSST = || Y^{c^{\prime}} - \mathbf{1}\overline{Y}_{\ln}^{c}||^{2} + || Y^{L^{\prime}} - \mathbf{1}\overline{Y}_{\ln}^{L}||^{2} + || Y^{R^{\prime}} - \mathbf{1}\overline{Y}_{\ln}^{R}||^{2}
\end{equation}
\begin{align*}
|| Y^{c^{\prime}} - \mathbf{1}\overline{Y}_{\ln}^{c}||^{2} & =|| Y^{c^{\prime}} - Y^{c^{\ast}} + Y^{c^{\ast}} - \mathbf{1}\overline{Y}_{\ln}^{c}||^{2}\\
& = \big(Y^{c^{\prime}}- Y^{c^{\ast}}\big)^{T}\big(Y^{c^{\prime}}- Y^{c^{\ast}}\big) + 2\big(Y^{c^{\prime}}- Y^{c^{\ast}}\big)^{T}\big(Y^{c^{\ast}} - \mathbf{1}\overline{Y}_{\ln}^{c}\big) + \big(Y^{c^{\ast}} - \mathbf{1}\overline{Y}_{\ln}^{c}\big)^{T}\big(Y^{c^{\ast}} - \mathbf{1}\overline{Y}_{\ln}^{c}\big)
\end{align*}
But
\begin{align*}
2\big(Y^{c^{\prime}}- Y^{c^{\ast}}\big)^{T}\big(Y^{c^{\ast}} - \mathbf{1}\overline{Y}_{\ln}^{c}\big) & = 2\big(Y^{c^{\prime}}- Y^{c^{\ast}}\big)^{T}Y^{c^{\ast}} - 2\big(Y^{c^{\prime}}- Y^{c^{\ast}}\big)^{T}\mathbf{1}\overline{Y}_{\ln}^{c}\\
& = 0 \quad \text{(from proposition \ref{prop})}.
\end{align*}
\begin{equation}\label{24}
\Rightarrow || Y^{c^{\prime}} - \mathbf{1}\overline{Y}_{\ln}^{c}||^{2} = || Y^{c^{\prime}}- Y^{c^{\ast}}||^{2} + ||Y^{c^{\ast}} - \mathbf{1}\overline{Y}_{\ln}^{c}||^{2}
\end{equation}
\begin{align*}
|| Y^{L^{\prime}} - \mathbf{1}\overline{Y}_{\ln}^{L}||^{2} & = || Y^{L^{\prime}} - Y^{L^{\ast}} + Y^{L^{\ast}} - \mathbf{1}\overline{Y}_{\ln}^{L}||^{2}\\
& = \big(Y^{L^{\prime}} - Y^{L^{\ast}}\big)^{T}\big(Y^{L^{\prime}} - Y^{L^{\ast}}\big) + \big(Y^{L^{\ast}} - \mathbf{1}\overline{Y}_{\ln}^{L}\big)^{T}\big(Y^{L^{\ast}} - \mathbf{1}\overline{Y}_{\ln}^{L}\big) + 2\big(Y^{L^{\prime}} - Y^{L^{\ast}}\big)^{T}\big(Y^{L^{\ast}} - \mathbf{1}\overline{Y}_{\ln}^{L}\big)
\end{align*}
\begin{align*}
2\big(Y^{L^{\prime}} - Y^{L^{\ast}}\big)^{T}\big(Y^{L^{\ast}} - \mathbf{1}\overline{Y}_{\ln}^{L}\big) & = 2\big(Y^{L^{\prime}} - Y^{L^{\ast}}\big)^{T}Y^{L^{\ast}} - 2\big(Y^{L^{\prime}} - Y^{L^{\ast}}\big)^{T}\mathbf{1}\overline{Y}_{\ln}^{L}\\
& = 0 \quad \text{(from proposition \ref{prop})}.
\end{align*}
\begin{equation}\label{25}
\Rightarrow || Y^{L^{\prime}} - \overline{Y}_{\ln}^{L}||^{2} = || Y^{L^{\prime}}- Y^{L^{\ast}}||^{2} + ||Y^{L^{\ast}} - \mathbf{1}\overline{Y}_{\ln}^{L}||^{2}
\end{equation}
\begin{align*}
|| Y^{R^{\prime}} - \mathbf{1}\overline{Y}_{\ln}^{R}||^{2} & = || Y^{R^{\prime}} - Y^{R^{\ast}} + Y^{R^{\ast}} - \mathbf{1}\overline{Y}_{\ln}^{R}||^{2}\\
& = \big(Y^{R^{\prime}} - Y^{R^{\ast}}\big)^{T}\big(Y^{R^{\prime}} - Y^{R^{\ast}}\big) + \big(Y^{R^{\ast}} - \mathbf{1}\overline{Y}_{\ln}^{R}\big)^{T}\big(Y^{R^{\ast}} - \mathbf{1}\overline{Y}_{\ln}^{R}\big) + 2\big(Y^{R^{\prime}} - Y^{R^{\ast}}\big)^{T}\big(Y^{R^{\ast}} - \mathbf{1}\overline{Y}_{\ln}^{R}\big)
\end{align*}
\begin{align*}
2\big(Y^{R^{\prime}} - Y^{R^{\ast}}\big)^{T}\big(Y^{R^{\ast}} - \mathbf{1}\overline{Y}_{\ln}^{R}\big) & = 2\big(Y^{R^{\prime}} - Y^{R^{\ast}}\big)^{T}Y^{R^{\ast}} - 2\big(Y^{R^{\prime}} - Y^{R^{\ast}}\big)^{T}\mathbf{1}\overline{Y}_{\ln}^{R}\\
& = 0 \quad \text{(from proposition \ref{prop})}.
\end{align*}
\begin{equation}\label{26}
\Rightarrow || Y^{R^{\prime}} - \mathbf{1}\overline{Y}_{\ln}^{R}||^{2} = || Y^{R^{\prime}}- Y^{R^{\ast}}||^{2} + ||Y^{R^{\ast}} - \mathbf{1}\overline{Y}_{\ln}^{R}||^{2}
\end{equation}
\begin{align}\nonumber
\eqref{23}-\eqref{24}-\eqref{25}-\eqref{26}\Rightarrow FSST & =|| Y^{c^{\prime}}- Y^{c^{\ast}}||^{2} + ||Y^{c^{\ast}} - \mathbf{1}\overline{Y}_{\ln}^{c}||^{2} + || Y^{L^{\prime}}- Y^{L^{\ast}}||^{2} \\
& + ||Y^{L^{\ast}} - \mathbf{1}\overline{Y}_{\ln}^{L}||^{2} + || Y^{R^{\prime}}- Y^{R^{\ast}}||^{2} + ||Y^{R^{\ast}} - \mathbf{1}\overline{Y}_{\ln}^{R}||^{2}\\
FSST & = FSSR + FSSE \quad \square .
\end{align}
\end{proof}
\begin{definition}\label{def2}
Let us consider a set of crisp observations $\mathbf{X}_{n\times (p+1)}$ and fuzzy output $\tilde{Y}_{ij} = (Y_{ij}^{L},Y_{ij}^{c}, Y_{ij}^{R})$ in the hybrid log-Poisson regression model \eqref{6} in loss reserving framework.\par
We define the fuzzy goodness of fit index by :
\begin{align}
\tilde{R}_{F}^{2} & = \dfrac{FSSR}{FSST}=1 - \dfrac{FSSE}{FSST}\\
                  & = \dfrac{|| Y^{c^{\ast}} - \mathbf{1}\overline{Y}_{\ln}^{c}||^{2} + || Y^{L^{\ast}} - \mathbf{1}\overline{Y}_{\ln}^{L}||^{2} + || Y^{R^{\ast}} - \mathbf{1}\overline{Y}_{\ln}^{R}||^{2}}{|| Y^{c^{\prime}} - \mathbf{1}\overline{Y}_{\ln}^{c}||^{2} + || Y^{L^{\prime}} - \mathbf{1}\overline{Y}_{\ln}^{L}||^{2} + || Y^{R^{\prime}} - \mathbf{1}\overline{Y}_{\ln}^{R}||^{2}}\\
                  & = 1 - \dfrac{|| Y^{c^{\prime}} - Y^{c^{\ast}}||^{2} + || Y^{L^{\prime}} - Y^{L^{\ast}}||^{2} + || Y^{R^{\prime}} - Y^{R^{\ast}}||^{2}}{|| Y^{c^{\prime}} - \mathbf{1}\overline{Y}_{\ln}^{c}||^{2} + || Y^{L^{\prime}} - \mathbf{1}\overline{Y}_{\ln}^{L}||^{2} + || Y^{R^{\prime}} - \mathbf{1}\overline{Y}_{\ln}^{R}||^{2}}
\end{align}
Using theorem \ref{prop2}, we notice that $\tilde{R}_{F}^{2} \in [0,1]$.
\end{definition}
\section{Algorithm for implementation of the new model}
In this section, we provide the algorithm behind our \texttt{R} program that will allow us to estimate the parameters of the new model and to compute the outstanding loss reserves.\par
\begin{description}
	\item[\textbf{1)Modelling the incremental losses $Y_{ij}$ with the log-Poisson regression :}] 
	We estimate the incremental losses $Y_{ij}$ through log-Poisson regression, i.e
	\begin{align*}
	& Y_{ij} \sim \mathcal{P}(e^{\nu_{ij}}), \quad\text{where}\; \nu_{ij}=\tau + \alpha_{i} + \gamma_{j}\\
	\Rightarrow & \hat{Y}_{ij} = e^{\mathbf{X}_{i}^{T}\hat{\beta}}, \; (i,j)\in \{1,\ldots ,k \}\times \{ 1,\ldots, k-i+1\}
	\end{align*}
	\item[\textbf{2) Estimation Procedure of $Y_{ij}^{R}, Y_{ij}^{c}, Y_{ij}^{L}$ :}]
	We assume that in the run off triangle $Y_{ij}$ have been modelled as follows :
	\begin{equation}
	Y_{ij} \sim \mathcal{P}(\varphi_{ij}), \;\varphi_{ij}=e^{\tau + \alpha_{i} + \gamma_{j}}
	\end{equation}  
	we compute now the Pearson residuals
    \begin{align*}
    & \hat{r}_{p_{ij}}=\dfrac{Y_{ij} - \hat{\varphi}_{ij}}{\sqrt{\hat{\varphi}_{ij}}}\\
    & \text{where}\;  \hat{\varphi}_{ij} = e^{\hat{\tau} + \hat{\alpha}_{i} + \hat{\gamma}_{j}} \; \text{and}\; (i,j)\in \{1,\ldots ,k \}\times \{ 1,\ldots, k-i+1\}
    \end{align*}
	The adjusted residuals (England and Verrall, 1999) are computed as follows :
	\begin{equation}
	\hat{r}^{\prime}_{p_{ij}}=\sqrt{\dfrac{n}{n-p}}\hat{r}_{p_{ij}}, \; \text{where}\; n=\dfrac{k(k+1)}{2}, \; p=2k-1
	\end{equation}
    Our idea to construct the fuzzy output $\tilde{Y}_{ij}$ is as following :
    \begin{itemize}
    	\item [$\bullet$] $Y^{c}_{ij} = Y_{ij}$
    	\item [$\bullet$] $Y^{L}_{ij} = Y_{ij} - \dfrac{|\hat{r}^{\prime}_{p_{ij}}|}{2}$
    	\item [$\bullet$] $Y^{R}_{ij} = Y_{ij} + \dfrac{|\hat{r}^{\prime}_{p_{ij}}|}{2}$
    \end{itemize}

	\item[\textbf{3) Estimation of fuzzy parameters in the hybrid log-Poisson model :}]
	We use the expression of  $\hat{\bm{\beta}}, \hat{\theta}, \hat{\delta}, \hat{\lambda}$ and $\hat{\mu}$ given in theorem \ref{pro1} to estimate the fuzzy parameters of the new model in equation \eqref{6}. For that, an iterative algorithm have been written under \texttt{R} to estimate those parameters.  
	\item[\textbf{4)Estimation of the goodness of fit :}] From step \textbf{3)}, we can get
	\begin{align}
	\hat{Y}^{c^{\ast}} & = \mathbf{X}\hat{\bm{\beta}}\\
	\hat{Y}^{L^{\ast}} & \mathbf{X}\hat{\bm{\beta}}\hat{\theta} + \mathbf{1}\hat{\lambda}\\
	\hat{Y}^{R^{\ast}} & \mathbf{X}\hat{\bm{\beta}}\hat{\delta} + \mathbf{1}\hat{\mu}
	\end{align}
	Then from definition \ref{def2}, the estimation of $\tilde{R}^{2}_{F}$ is given by 
	\begin{equation}
	\tilde{R}^{2}_{F} = 1 - \dfrac{|| Y^{c^{\prime}} - \hat{Y}^{c^{\ast}}||^{2} + || Y^{L^{\prime}} - \hat{Y}^{L^{\ast}}||^{2} + || Y^{R^{\prime}} - \hat{Y}^{R^{\ast}}||^{2}}{|| Y^{c^{\prime}} - \mathbf{1}\overline{Y}_{\ln}^{c}||^{2} + || Y^{L^{\prime}} - \mathbf{1}\overline{Y}_{\ln}^{L}||^{2} + || Y^{R^{\prime}} - \mathbf{1}\overline{Y}_{\ln}^{R}||^{2}}
	\end{equation}
	\item[\textbf{5)Estimation of outstandings reserves :}]
	Et this step, we predict $\tilde{Y}_{ij}$ using the new model  as follows :
	\begin{align}
	\hat{Y}_{ij}^{c^{\prime}} & = \mathbf{X}_{i}^{T}\hat{\bm{\beta}} \Rightarrow \hat{Y}_{ij}^{c} = e^{\mathbf{X}_{i}^{T}\hat{\bm{\beta}}}\\
	\hat{Y}_{ij}^{L^{\prime}} & = \mathbf{X}_{i}^{T}\hat{\bm{\beta}}\hat{\theta} + \hat{\lambda} \Rightarrow \hat{Y}_{ij}^{L} = e^{\mathbf{X}_{i}^{T}\hat{\bm{\beta}}\hat{\theta} + \hat{\lambda}}\\
	\hat{Y}_{ij}^{R^{\prime}} & = \mathbf{X}_{i}^{T}\hat{\bm{\beta}}\hat{\delta} + \hat{\mu} \Rightarrow \hat{Y}_{ij}^{R} = e^{\mathbf{X}_{i}^{T}\hat{\bm{\beta}}\hat{\delta} + \hat{\mu}}
	\end{align}
	where
	\begin{align*}
	(i,j) & \in \{1, \ldots , k\}\times \{ k-i+1, \ldots , k\}\\
	\hat{\bm{\beta}} & = \big( \hat{\tau}, \hat{\alpha}_{2}, \ldots , \hat{\alpha}_{k}, \hat{\gamma}_{2}, \ldots, \hat{\gamma}_{k}\big)^{T}\in \mathbb{R}^{2k-1}\\
	\hat{\tilde{Y}}_{ij} & = (\hat{Y}_{ij}^{L}, \hat{Y}_{ij}^{c}, \hat{Y}_{ij}^{R})
	\end{align*}
	Then the fuzzy total loss reserve is computed as follow:
    \begin{align}
    	\tilde{R}_{T.Res} & = \sum\limits_{i=1}^{k}\sum\limits_{j=k-i+1}^{k}\hat{\tilde{Y}}_{ij}\\
    	                  & =\big(\sum\limits_{i=1}^{k}\sum\limits_{j=k-i+1}^{k}\hat{Y}_{ij}^{L},\; \sum\limits_{i=1}^{k}\sum\limits_{j=k-i+1}^{k}\hat{Y}_{ij}^{c},\; \sum\limits_{i=1}^{k}\sum\limits_{j=k-i+1}^{k}\hat{Y}_{ij}^{R}\big)\\
    	                  & = \big(R^{L}_{T.Res}, R^{c}_{T.Res},R^{R}_{T.Res}\big)
    \end{align}
In this article, we use the concept of expected value of FN (de Campos Ib\'{a}\~{n}ez and Mu\~{n}oz, 1989) to move from the fuzzy value of total loss reserve $\tilde{R}_{T.Res}$ to the crisp value of total loss reserve $R_{T.Res}$. Denote $\mathbb{E}_{F}$ that expected value. \par
The $h-$level of fuzzy total loss reserve is defined as following :
\begin{equation}
\tilde{R}_{T.Res}(h)=\big[h\cdot R^{c}_{T.Res} - (1-h)\cdot R^{L}_{T.Res};h\cdot R^{c}_{T.Res} + (1-h)\cdot R^{R}_{T.Res}\big]
\end{equation}
Then the expected value of FN $\tilde{R}_{T.Res}$ is defined as follows :
\begin{equation}\label{29}
\mathbb{E}_{F}(\tilde{R}_{T.Res}, \pi) = (1-\pi)\int_{0}^{1}\big(h\cdot R^{c}_{T.Res} - (1-h)\cdot R^{L}_{T.Res}\big) dh + \pi\int_{0}^{1}\big(h\cdot R^{c}_{T.Res} + (1-h)\cdot R^{R}_{T.Res}\big) dh 
\end{equation}
where $\pi$ is the decision-maker risk aversion parameter $(0 \leqslant \pi \leqslant 1)$. From \eqref{29}, we have
\begin{align*}
\mathbb{E}_{F}(\tilde{R}_{T.Res}, \pi) & = (1-\pi)\int_{0}^{1}\big( h\cdot R^{c}_{T.Res} - R^{L}_{T.Res} + h\cdot R^{L}_{T.Res}\big)dh + \pi \int_{0}^{1}\big( h\cdot R^{c}_{T.Res} + R^{R}_{T.Res} - h\cdot R^{R}_{T.Res}\big)dh\\
& = (1-\pi)\int_{0}^{1}h(R^{c}_{T.Res} + R^{L}_{T.Res})dh - (1-\pi)\int_{0}^{1}R^{L}_{T.Res} dh + \pi\int_{0}^{1}h(R^{c}_{T.Res} - R^{R}_{T.Res})dh\\
& + \pi\int_{0}^{1}R^{R}_{T.Res} dh\\
& = \dfrac{(1-\pi)(R^{c}_{T.Res} + R^{L}_{T.Res})}{2} - (1-\pi)R^{L}_{T.Res} + \dfrac{\pi (R^{c}_{T.Res} - R^{R}_{T.Res})}{2} + \pi R^{R}_{T.Res}\\
& = \dfrac{(1-\pi)(R^{c}_{T.Res} - R^{L}_{T.Res})}{2} + \dfrac{\pi (R^{c}_{T.Res} + R^{R}_{T.Res})}{2}
\end{align*}

\end{description}
\section{Numerical Example}
In this section, we apply both the classical log-Poisson regression (Mack, 1991) and and the hybrid model estimated by a fuzzy least square procedure on a real data. Let us use the numerical example from de Andr\'{e}s S\'{a}nchez (2006).\par
\begin{table}[!h]
	\centering
	\begin{tabular}{|c|c|ccccc|}
		\cline{3-7}
		\multicolumn{2}{c|}{} &	\multicolumn{5}{|c|}{\centering Development Year}\\ \hline
		\multicolumn{2}{|c|}{$i/j$} &  0 &  1 &  2 &  3 &  4 \\ \hline
		\multirow{4}{3mm}{\rotatebox{90}{\bf Origin Year}} & 2000  & 1120 & 2090 & 2610 & 2920 & 3130 \\
		&  2001  & 1030 & 1920 & 2370 & 2710 & \\
		&  2002  & 1090 & 2140 & 2610 &  &  \\
		&  2003  & 1300 & 2650 & & & \\
		&  2004  & 1420  & & & & \\ \hline
	\end{tabular}
	\caption{Numerical example from de Andr\'{e}s S\'{a}nchez (2006)}\label{tab2}
\end{table}
According to the $\mathbf{1^{st}}$ step of the algorithm, we perform the classical log-Poisson regression on the data from de Andr\'{e}s S\'{a}nchez (2006) using \textbf{R} software. The estimated parameters are displayed in table \ref{tab3}\\
\begin{table}[!h]
	\centering
	\begin{tabular}{|c|c|c|c|}
		\hline
	 $(\hat{\alpha}_{i})_{2 \leqslant i \leqslant 5}$ & $(\hat{\gamma}_{j})_{2 \leqslant j \leqslant 5}$  & $p-\text{value} (\hat{\alpha}_{i})$ & $p-\text{value}(\hat{\gamma}_{j})$  \\ \hline \hline
     $-0.08473$ & $0.66182$ & $4.23 \times e^{-08}$ & $< 2 \times e^{-16}$ \\
      $0.00587$ & $0.86503$ & $0.741$    & $< 2 \times e^{-16}$ \\
      $0.20725$ & $0.98780$ & $< 2\times e^{-16}$  & $< 2 \times e^{-16}$ \\
      $0.26203$ & $1.05240$ & $4.72 \times e^{-16}$ & $< 2 \times e^{-16}$ \\\hline
	  \multicolumn{2}{|c}{ $\hat{\tau}= 6.99639$} & \multicolumn{2}{c|}{ $p-\text{value}(\hat{\tau}) =< 2 \times e^{-16} $}\\\hline 
	  \multicolumn{2}{|c|}{$\mathbf{R^{2}} = 0.9621253$} & \multicolumn{2}{c|}{$\text{\textbf{Total Reserve}} = 33634.89$}\\ \hline
	\end{tabular}
	\caption{Estimated parameters}\label{tab3}
\end{table} \\
From table \ref{tab3} and With a threshold of $1\%$, we conclude that except $\hat{\alpha}_{3}$, the others coefficients are statistically significant. The goodness of fit of the model to the data is good, since $ \mathbf{R^2}=\mathbf{96.21\%}$ and the estimation of the total loss reserve is $\mathbf{33634.89}$\par
Now let us test if the model performed is adapted to a statistical perspective through a dispersion test (see table \ref{tab4}).\\
\begin{table}[!h]
	\centering
	\begin{tabular}{|c|}
		\hline
		 Overdispersion test  \\ \hline \hline
		 $H_{0} : \; \psi = 1$ \\ 
		 $H_{1} : \; \psi > 1$ \\ 
		 $Z=-6.3414$                  \\ 
		 $p-\text{value} = 1$    \\ \hline
	\end{tabular}
	\caption{Overdispersion test}\label{tab4}
\end{table}
With a threshold of $1\%$, we do not reject the null hypothesis, i.e $p-value > 1\%$. Therefore we don't need to perform a quasi-Poisson regression.\par
Let us perform the steps \textbf{2),3)} and \textbf{4)} of the algorithm. The iterative algorithm to estimate parameters $\bm{\beta}, \theta, \delta, \lambda, \mu$ converges after $12112$ iterations and the estimation of parameters and fuzzy output are displayed in table \ref{tab5} and in equation \eqref{28}.\\
\begin{table}[!h]
	\centering
	\begin{tabular}{|c|c|c|}
		\hline
	$\hat{Y}_{11}^{L}=7.002797$	& $\hat{Y}_{11}^{c}=7.003261$ & $\hat{Y}_{11}^{R}=7.003733$  \\  \hline
	$\hat{Y}_{12}^{L}=7.661440$	& $\hat{Y}_{12}^{c}=7.661621$ & $\hat{Y}_{12}^{R}=7.661801$  \\  \hline
	$\hat{Y}_{13}^{L}=7.862986$	& $\hat{Y}_{13}^{c}=7.863080$ & $\hat{Y}_{13}^{R}=7.863171$  \\  \hline
	$\hat{Y}_{14}^{L}=7.984444$	& $\hat{Y}_{14}^{c}=7.984487$ & $\hat{Y}_{14}^{R}=7.984523$  \\  \hline
	$\hat{Y}_{15}^{L}=8.048776$	& $\hat{Y}_{15}^{c}=8.048791$ & $\hat{Y}_{15}^{R}=8.048798$  \\  \hline
	$\hat{Y}_{21}^{L}=6.917834$	& $\hat{Y}_{21}^{c}=6.918335$ & $\hat{Y}_{21}^{R}=6.918845$  \\  \hline
	$\hat{Y}_{22}^{L}=7.576477$	& $\hat{Y}_{22}^{c}=7.576694$ & $\hat{Y}_{22}^{R}=7.576912$  \\  \hline
	$\hat{Y}_{23}^{L}=7.778023$	& $\hat{Y}_{23}^{c}=7.778154$ & $\hat{Y}_{23}^{R}=7.778282$  \\  \hline
	$\hat{Y}_{24}^{L}=7.899481$	& $\hat{Y}_{24}^{c}=7.899560$ & $\hat{Y}_{24}^{R}=7.899634$  \\  \hline
	$\hat{Y}_{31}^{L}=7.003343$	& $\hat{Y}_{31}^{c}=7.003807$ & $\hat{Y}_{31}^{R}=7.004279$  \\  \hline
	$\hat{Y}_{32}^{L}=7.661985$	& $\hat{Y}_{32}^{c}=7.662166$ & $\hat{Y}_{32}^{R}=7.662346$  \\  \hline
	$\hat{Y}_{33}^{L}=7.863531$	& $\hat{Y}_{33}^{c}=7.863626$ & $\hat{Y}_{33}^{R}=7.863716$  \\  \hline
	$\hat{Y}_{41}^{L}=7.196656$	& $\hat{Y}_{41}^{c}=7.197037$ & $\hat{Y}_{41}^{R}=7.197423$  \\  \hline
	$\hat{Y}_{42}^{L}=7.855299$	& $\hat{Y}_{42}^{c}=7.855397$ & $\hat{Y}_{42}^{R}=7.855490$  \\  \hline
	$\hat{Y}_{51}^{L}=7.258057$	& $\hat{Y}_{51}^{c}=7.258411$ & $\hat{Y}_{51}^{R}=7.258770$  \\ \hline
	\end{tabular}
	\caption{Estimation of output from the hybrid model}\label{tab5}
\end{table}
\begin{align}\label{28}
\hat{\bm{\beta}}=\begin{pmatrix}
 7.0032610203 \\
 -0.0849264367 \\
 0.6583597871 \\
 0.0005455323 \\
 0.8598194144 \\
 0.1937759436\\
 0.9812255202 \\
 0.2551496695 \\
 1.0455297360
\end{pmatrix}
& \quad\hat{\theta} = 1.000429; & \quad\hat{\lambda}=-0.003468438;\quad & \hat{\delta} = 0.9995556;\quad & \hat{\mu} = 0.003584175;\\
\quad & \tilde{R}_{F}^{2}=0.9986105 \nonumber
\end{align}
From these outputs, we conclude that the hybrid model is more adequate to the classical one since $\tilde{R}_{F}^{2} > R^{2}$.\par
From step \textbf{5)} of our algorithm, we can predict the incremental losses as fuzzy numbers and total fuzzy loss reserve $\mathbf{\tilde{R}_{T.Res}}$. The results are given in table \ref{tab6}.\\
\begin{table}[!h]
	\centering
	\begin{tabular}{|c|c|c|}
		\hline
		$\hat{Y}_{25}^{L}= 2875.014$	& $\hat{Y}_{25}^{c}= 2875.162$ & $\hat{Y}_{25}^{R}= 2875.293$  \\  \hline
		$\hat{Y}_{34}^{L}= 2936.547$	& $\hat{Y}_{34}^{c}= 2936.671$ & $\hat{Y}_{34}^{R}= 2936.777$  \\  \hline
		$\hat{Y}_{35}^{L}= 3131.669$	& $\hat{Y}_{35}^{c}= 3131.716$ & $\hat{Y}_{35}^{R}= 3131.739$  \\  \hline
		$\hat{Y}_{43}^{L}= 3155.319$	& $\hat{Y}_{43}^{c}= 3155.355$ & $\hat{Y}_{43}^{R}= 3155.368$  \\  \hline
		$\hat{Y}_{44}^{L}= 3562.481$	& $\hat{Y}_{44}^{c}= 3562.659$ & $\hat{Y}_{44}^{R}= 3562.803$  \\  \hline
		$\hat{Y}_{45}^{L}= 3798.981$	& $\hat{Y}_{45}^{c}= 3799.279$ & $\hat{Y}_{45}^{R}= 3799.538$  \\  \hline
		$\hat{Y}_{52}^{L}= 2742.702$	& $\hat{Y}_{52}^{c}= 2742.898$ & $\hat{Y}_{52}^{R}= 2743.080$  \\  \hline
		$\hat{Y}_{53}^{L}= 3355.000$	& $\hat{Y}_{53}^{c}= 3355.077$ & $\hat{Y}_{53}^{R}= 3355.127$  \\  \hline
		$\hat{Y}_{54}^{L}= 3787.870$	& $\hat{Y}_{54}^{c}= 3788.162$ & $\hat{Y}_{54}^{R}= 3788.415$  \\  \hline
		$\hat{Y}_{55}^{L}= 4039.332$	& $\hat{Y}_{55}^{c}= 4039.759$ & $\hat{Y}_{55}^{R}= 4040.141$  \\  \hline
		\multicolumn{3}{|c|}{ $\mathbf{\tilde{R}_{T.Res}= (33384.915, 33386.738, 33388.281)}$} \\  \hline
	\end{tabular}
	\caption{Predicted values from the hybrid model}\label{tab6}
\end{table}
From the table \ref{tab6} and using the expected value of fuzzy number for defuzzification purposes, we can compute the crisp value of outstanding loss reserve with the maximum decision-maker risk aversion, i.e $\pi = 1$.
\begin{align*}
\mathbb{E}_{F}(\tilde{R}_{T.Res}, \pi=1)= \hat{R}_{T.Res} & = \dfrac{R^{c}_{T.Res} + R^{R}_{T.Res}}{2}\\
                                                          & = \mathbf{33387.5095}
\end{align*}
From those results we conclude that the new hybrid model we suggested produce best results than the classical one according to the goodness of fit.
\section{Conclusion}
This paper has considered the relevance of Hybrid Models in loss reserving framework, mainly when we are in presence of vague information like in medical insurance (Straub and Swiss, 1988). Those models could give best result compared to stochastic models. In our previous article, we have estimated the hybrid log-Poisson model using a linear programming problem and a numerical example have been made in view to compare that model with the classical log-Poisson regression. In this article, we have suggested a new way to estimate the parameters of the hybrid log-Poisson regression in loss reserving framework using the fuzzy least squares procedure. Furthermore we have developed a goodness of fit index to assess our model. This new model have been applied to a \textit{run-off triangle} in order to estimate the outstanding loss reserve. According to the goodness of fit, the hybrid model approaches the fair value of loss reserve better than the well known log-Poisson regression model (Mack, 1991). However the weakness of that fuzzy least squares estimation of the hybrid log-Poisson regression is its computational part. Since we got an iterative estimator, the \textbf{R} program take some time to converge (\textbf{12112} iterations).
\section*{Acknowledgements}
This work was supported by African Union through Pan African University/Jomo Kenyatta University of Agriculture and Technology. \newpage

\textbf{{\large References :}}\\

Asai, H. T.-S. U.-K. (1982). Linear regression analysis with fuzzy model. IEEE Trans. Systems Man
Cybern, 12:903-907.\par
Bornhuetter, R. L. and Ferguson, R. E. (1972). The actuary and ibnr. In Proceedings of the casualty
actuarial society, volume 59, pages 181-195. \par
Buckley, J. J. (2006). Fuzzy probability and statistics, volume 196. Springer Science \& Business
Media. \par
Celmiņ\v{s}, A. (1987a). Least squares model fitting to fuzzy vector data. Fuzzy sets and systems, 22(3):245-269. \par
Celmiņ\v{s}, A. (1987b). Multidimensional least-squares fitting of fuzzy models. Mathematical Modelling,
9(9):669-690.\par
de Andrés Sánchez, J. (2006). Calculating insurance claim reserves with fuzzy regression. Fuzzy sets
and systems, 157(23):3091-3108.\par
de Andr\'{e}s-S\'{a}nchez, J. (2007). Claim reserving with fuzzy regression and taylors geometric separation method. Insurance: Mathematics and Economics, 40(1):145-163.\par
de Andr\'{e}s-S\'{a}nchez, J. (2012). Claim reserving with fuzzy regression and the two ways of anova.
Applied Soft Computing, 12(8):2435-2441.\par
de Andr\'{e}s S\'{a}nchez, J. (2014). Fuzzy claim reserving in non-life insurance. Comput. Sci. Inf. Syst., 11(2):825-838.\par
de Campos Ib\'{a}\~{n}ez, L. M. and Mu\~{n}oz, A. G. (1989). A subjective approach for ranking fuzzy
numbers. Fuzzy sets and systems, 29(2):145-153.\par
Dubois, D. and Prade, H. (1978). Operations on fuzzy numbers. International Journal of systems
science, 9(6):613-626.\par
Dubois, D. and Prade, H. (1988). Fuzzy numbers : An overview. In Analysis of Fuzzy Information,
pages 3-39. J. C Bezdek.\par
D'Urso, P. (2003). Linear regression analysis for fuzzy/crisp input and fuzzy/crisp output data.
Computational Statistics \& Data Analysis, 42(1-2):47-72.\par
D'Urso, P. and Gastaldi, T. (2000). A least-squares approach to fuzzy linear regression analysis.
Computational Statistics \& Data Analysis, 34(4):427-440.\par
DUrso, P. and Gastaldi, T. (2001). Linear fuzzy regression analysis with asymmetric spreads. In
Advances in Classification and Data Analysis, pages 257-264. Springer.\par
England, P. and Verrall, R. (1999). Analytic and bootstrap estimates of prediction errors in claims
reserving. Insurance: mathematics and economics, 25(3):281-293.\par
England, P. D. and Verrall, R. J. (2002). Stochastic claims reserving in general insurance. British
Actuarial Journal, 8(03):443-518.\par
Ishibuchi, H. and Nii, M. (2001). Fuzzy regression using asymmetric fuzzy coefficients and fuzzified
neural networks. Fuzzy Sets and Systems, 119(2):273-290.\par
Kauffman, A. and Gupta, M. M. (1991). Introduction to fuzzy arithmetic, theory and application.
Lai, Y.-J. and Hwang, C.-L. (1992). Fuzzy mathematical programming. In Fuzzy Mathematical
Programming, pages 74-186. Springer.\par
Linnemann, P. (1984). van eeghen j. (1981): Loss reserving methods. surveys of actuarial studies
no. 1. nationale-nederlanden n.v., rotterdam. 114 pages. ASTIN Bulletin: The Journal of the
International Actuarial Association, 14(01):87-88.\par
Mack, T. (1991). A simple parametric model for rating automobile insurance or estimating ibnr
claims reserves. Astin bulletin, 21(01):93-109.\par
Straub, E. and Swiss, A. A. (1988). Non-life insurance mathematics. Springer.\par
Taylor, G. (1986). Claims reserving in non-life insurance. Insurance series. North-Holland.\par
Taylor, G., McGuire, G., and Greenfield, A. (2003). Loss reserving: past, present and future.\par
W\"{u}thrich, M. V. and Merz, M. (2008). Stochastic claims reserving methods in insurance, volume
435. John Wiley \& Sons.\par
Yager, R. R. and Filev, D. (1999). On ranking fuzzy numbers using valuations. International Journal
of Intelligent Systems, 14(12):1249-1268.\par
Zadeh, L. A. (1965). Fuzzy sets. Information and control, 8(3):338-353.
\end{document}